\documentclass[12pt, draftclsnofoot, onecolumn]{IEEEtran}
\usepackage{color}
\usepackage{amssymb}
\usepackage{cite}
\usepackage[cmex10]{amsmath}
\usepackage{algorithm}
\usepackage{algorithmic} 
\usepackage{array}
\usepackage{mathrsfs}
\usepackage{graphicx}
\usepackage{latexsym}
\usepackage{amscd}  
\usepackage{amsfonts}
\usepackage{subfigure}
\usepackage{amsmath,amscd,amssymb,verbatim}
\usepackage{graphics}
\usepackage{amsthm}
\usepackage[T1]{fontenc}
\usepackage[utf8]{inputenc}
\usepackage{authblk}
\usepackage{amsmath,amsthm}

\begin{document}
\title{Optimal Scheduling  of Reliability-Constrained   Relaying  System under Outdated CSI in the Finite Blocklength Regime }
\author{Yulin Hu \emph{IEEE Member},  Anke Schmeink,  \emph{IEEE Member} and James Gross, \emph{IEEE Senior Member}}

\maketitle

\begin{abstract}
Under the assumption of outdated channel state information (CSI) at the source, we consider the finite blocklength (FBL) throughput of a two-hop relaying system. 
Previous work has considered this setting so far only for the infinite blocklength case, where decoding can be arbitrarily reliable as long as operating below the Shannon limit. 
In contrast, in the FBL regime residual decoding errors can not be avoided even when transmitting below the Shannon limit.
This makes the scheduling problem at the source more vulnerable to transmission errors, where we investigate the trade-off between the choice of so called scheduling weights to avoid transmission errors and the resulting coding rate. 
We show that the corresponding maximization of the throughput under a reliability constraint can be solved efficiently by iterative algorithms. 
Nevertheless, the optimal solution requires a recomputation of the scheduling weights prior to each transmission. 
Thus, we also study heuristics relying on choosing the scheduling weights only once. 
Through numerical analysis, we first provide insights on the structure of the throughout under different scheduling weights and channel correlation coefficients. 
We then turn to the comparison of the optimal scheduling with the heuristic and show that the performance gap between them is only significant for relay systems with high average signal-to-noise ratios (SNR) on the backhaul and relaying link.  
In particular, the optimal scheduling scheme provides most value in case that the data transmission is subject to strict reliability constraints, justifying the significant additional computational burden.
\end{abstract}

\begin{IEEEkeywords}
 decode-and-forward, finite blocklength, optimal scheduling, outdated CSI, relaying.
\end{IEEEkeywords}

\IEEEpeerreviewmaketitle

\section{Introduction}

In wireless communications,  relaying~\cite{Laneman_2004,  Karmakar_2011, Li_2014} is 
well known as an efficient way to mitigate fading by exploiting spatial diversity and providing better channel quality.  
Specifically, two-hop decode-and-forward (DF) relaying protocols significantly improve the throughput and quality of service~\cite{Wendong_2011, Yulin_2011,   Hu_2015_effective_capacity,Bhatnagar_2013}.  
{However, typically these studies on the advantage of relaying are under the ideal assumption of communicating arbitrarily reliable at Shannon's channel capacity, i.e.,  code words are assumed to be infinitely long.}

In the finite blocklength regime, the data transmission is no longer arbitrarily reliable. Especially when the blocklength is short, the error probability (due to noise) becomes significant even if the rate is selected below the Shannon limit. 
Taking this into account, an accurate approximation of the achievable coding rate under the finite blocklength assumption for an additive white Gaussian noise (AWGN) channel was derived in~\cite{Verdu_2010} for a single-hop transmission system.  
Subsequently, the initial work for AWGN channels was extended to Gilbert-Elliott channels~\cite{Polyanskiy_2011} as well as quasi-static fading channels~\cite{Yang_2014,Gursoy_2013,Gursoy_2013_2,Peng_2011,Makki_2014,Xu_2016}. 
 It is shown in these works 
 that the finite blocklength performance of a single-hop transmission is determined by the  coding rate,  error probability and blocklength. 
 In particular, the performance loss due to the additional decoding errors at finite blocklength is considerable and increases as the blocklength decreases.  
 Also, if the channel and the blocklength are given, the error probability of the single-hop transmission is strictly increasing in the coding rate.
In our own previous work~\cite{Hu_letter_2015,Hu_2015,Hu_TWC_2015}, we extended   Polyanskiy's model~\cite{Verdu_2010} of single-hop transmission  to  a two-hop DF relaying network, where the relay halves the distance to provide a power gain but at the same time also halves the blocklength of the transmission. 
Subsequently, we provided a general analytical model of the finite blocklength performance under static/quasi-static channels in~\cite{Hu_letter_2015,Hu_2015,Hu_TWC_2015} while assuming the transmitter to have only average CSI. 
More recently, the throughput  of a relaying network with finite blocklengths and queuing constraints was studied in~\cite{Li_2016_ISIT} under the perfect CSI assumption. 

In practical relay systems (as for instance specified by the LTE standard) CSI feedback mechanisms are usually implemented, i.e., allowing the receiver to instantaneously estimate and feedback the CSI to the transmitter.   
However, typically there exists a delay between the instant of sampling the channel and the point in time when this CSI sample is received by the transmitter making the CSI feedback delayed and therefore outdated. 
The performance analysis and optimization of relaying systems operating on outdated CSI have been widely discussed in the infinite blocklength (IBL) regime. 
In~\cite{Hyadi_TVT_2015}, the probability of an outage event (defined as the event when the coding rate is higher than the Shannon capacity) of a DF relaying network is studied under the outdated CSI relaying scenario. 
Protocols are designed in~\cite {Jiang_TWC_2016} for a relay system operating based on outdated CSI to optimally trade-off outage, delay, and throughput.  
For multi-relay scenarios with outdated CSI, optimal relay selection algorithms~\cite{Fei_IET_2016, Michalopoulos_2016} are proposed to minimize the outage probability. 
However, these works generally ignore the impact of transmitting under finite blocklength restrictions, which introduces further subtleties in addition to the imperfect channel knowledge.


In this paper, we thus study the finite blocklength performance of a relaying network assuming the source to have only outdated CSI. 
Different from the average CSI scenario considered in~\cite{Hu_2015,Hu_letter_2015,Hu_TWC_2015},  based on the provided CSI the source is able to adjust the coding rate per frame. 
However, due to the outdated CSI, scheduling the coding rate in this way can result in more frequent transmission errors. 
We hence study the optimal scheduling of the coding rate in the relay system with outdated CSI under a reliability constraint for the data transmission. 
As objective function we focus on the maximization of the FBL throughput.
Solving this optimal scheduling problem requires the source to choose the coding rate based on scheduling weights, i.e. factors by which the outdated channel SNRs are rescaled.  
We contribute by first deriving a model for the FBL throughput of the relaying system operating with outdated CSI.
Next, we propose an optimal scheduling scheme that maximizes the FBL throughput. 
We show that the objective function of the scheduling problem is concave in the coding rate and quasi-concave in the scheduling weights. 
Therefore, the optimal scheduling problem can be solved efficiently by iterative methods.
Nevertheless, to mitigate the computational complexity, we also consider a sub-optimal scheduling scheme, where   fixed scheduling weights are applied per frame.
We refer to this scheme as constant heuristic and study the problem of choosing the constant scheduling weights which maximize the average FBL throughput over time. 
We finally perform numerical evaluations and show that the optimal scheme outperforms the best constant heuristic, especially when the reliability constraint is strict and/or the average SNR is high. 
Surprisingly, we find that the channel correlation has only a marginal impact on the performance gap between the two schemes.
 
The rest of the paper is organized as follows. Section~\ref{sec:Preliminaries} describes the system model and briefly reviews the background regarding the finite blocklength regime. 
In Section~\ref{sec:Physical_layer_P}, we first derive the finite blocklength performance model of the considered relaying scenario with outdated CSI. 
Afterwords, we state the optimization problem of interest and provide the theoretical insights that lead to the optimal solution. 
We then turn to the constant heuristic and provide an optimal solution for choosing the fixed scheduling weight. 
In Section~\ref{sec:simulation} we then present our numerical results. 
Finally, we conclude our work in Section~\ref{sec:Conclusion}.

\section{System Model}
\label{sec:Preliminaries}
%
We consider a  straightforward  scenario with a source $\rm{S}$, a destination~$\rm{D}$ and a  relay~$\rm{R}$ as schematically shown in Figure~\ref{system-instr}.
\begin{figure}[h]
\centering{\includegraphics[width=100mm, trim=10 10 20 15]{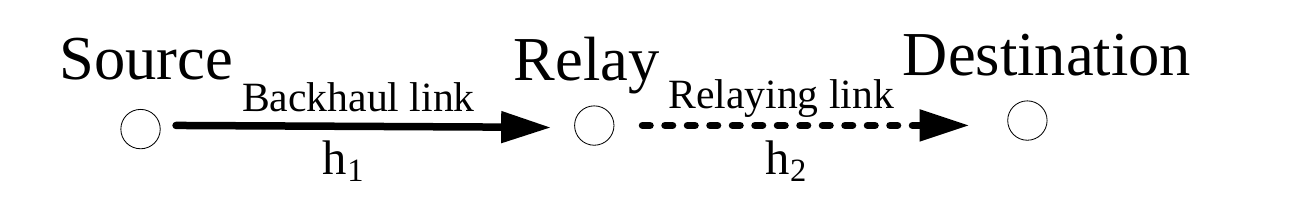}}
\caption{Example of the considered DF relaying scenario.}
\label{system-instr}
 
\end{figure}
The relay is assumed to work under a DF principle. 
The entire system operates in a slotted fashion  where time is divided into frames of length $n+2m$ symbols,  as shown in  Figure~\ref{system-instr2}. 
\begin{figure}[h]
\centering{\includegraphics[width=110mm, trim=15 5 25 10]{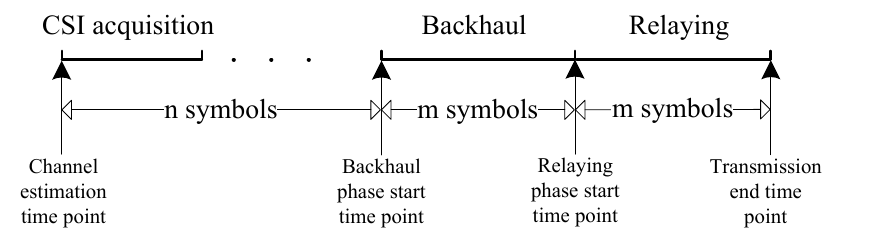}}
\caption{Illustration of the relationship of channel estimation and data transmission within the considered frame time.}
\label{system-instr2}
\end{figure}
Each frame consists of two parts, the initialization part and the transmission part. 
For the initialization part a certain amount of symbols are spent for acquiring the CSI. 
During this part, messages are exchanged to essentially obtain the CSI of the backhaul and relaying link at the source node.
We assume this part to have a duration of $n$ symbols, without specifying closer the exact system operation.
The second part of each frame is the transmission part containing two phases, which are the backhaul phase (of length $m$) and the relaying phase (of length $m$).
During the backhaul phase, the source sends a data block to the relay. Then, if the relay decodes the block successfully,  it forwards the block to the destination in the subsequent relaying phase.
Overall, we assume a setting where the initialization part takes a significantly longer amount of time than a single data transmission phase, i.e. $n > m$.
This is motivated for example by rather short (but important) payload packets for which a reliable transmission is crucial, which justifies the acquisition of the CSI upfront.

Channels are assumed to experience a time-varying Rayleigh-distributed random fading. As both the backhaul phase and the relaying phase are short, we assume that the channel state is constant during each phase. 
However, the channel states in different frames are assumed to be independent.
Considering a frame~$i$, the channel's complex states of the backhaul link and the relaying link are denoted by $h_{1,i}$ and $h_{2,i}$ and are assumed to be independent and identically distributed  (i.i.d.). 
The received SNR at the relay of the backhaul phase and the received SNR at the destination of the relaying phase are denoted by ${\gamma _{1,i}}$ and ${\gamma _{2,i}}$.  Hence, we have ${\gamma _{k,i}}= \bar \gamma_k {h^2_{k,i}}, k = 1,2$, where  $ \bar \gamma_k$ is the average SNR of link~$k$ (either the backhaul link or the relaying link). 
Recall that we assume the source to acquire the instantaneous CSI by sampling the channel $n$ symbols prior to the backhaul phase and $n+m$ symbols prior to the relaying phase. 
Thus, due to the time-varying nature of the fading, the sampled channel coefficients, denoted by  ${\hat h _{k,i}}, k =1,2$,  differ from the actual instantaneous channel coefficients $h_{k,i}$  that the data packet will experience. We adopt the widely-used Jakes model for the relation between ${\hat h _{k,i}}$ and $h_{k,i}$~\cite{Mallik_2003,Vicario_2006}:
\begin{equation}
\label{eq:Relationship_original}
{h_{k,i}} = {\rho _k}{\hat h_{k,i}} + \sqrt {1 - {\rho _k}^2} {e_{k,i}},
\end{equation}
where ${e_{k,i}}$ is a complex Gaussian random variable, i.e., ${e_{k,i}} \sim  {\mathcal {CN}} (0,1)$. In addition,   $\rho _k , k = 1,2$ are channel correlation coefficients. 
Taking the frame sequence into account, we thus obtain ${\rho _1} = {J_0}(2\pi {f_{{\rm{S - R}}}}n)$ and ${\rho _2} = {J_0}(2\pi {f_{{\rm{R - D}}}}(n + m))$, where ${f_{{\rm{S - R}}}}$ and  ${f_{{\rm{R - D}}}}$ stand for the Doppler frequency experienced on the backhaul link and the relaying link. In addition,  $J_0(\cdot)$ denotes the zero-order Bessel function of the first kind~\cite{Bessel_2005}. Based on the outdated CSI ${\hat h} _{k,i}$, the outdated SNRs are given by ${\hat\gamma _{k,i}}=  \bar \gamma_k {{\hat h}^2 _{k,i}}, k =1,2$.
Thus, the instantaneous channel SNRs ${\gamma _{k,i}}$  become now random variables conditioned on the outdated SNRs ${\hat \gamma _{k,i}}$. 
The conditional probability density function (PDF) of the instantaneous SNRs of link $k$ during frame $i$ thus results to~\cite{Mallik_2003}:
\begin{equation} 
\label{eq:Bessel_probability}
\mathbb{P}\left[ \gamma _{k,i} | \hat{\gamma} _{k,i}\right]= \frac{{\exp ( - \frac{{{\gamma _{k,i}} + \rho _k^2{{\hat \gamma }_{k,i}}}}{{{{\bar \gamma }_{k}}(1 - \rho _k^2)}})}}{{{{\bar \gamma }_{k}}(1 - \rho _k^2)}} \cdot {I_0}\left( {\frac{{2{\rho _k}\sqrt {{\gamma _{k,i}}{{\hat \gamma }_{k,i}}} }}{{{{\bar \gamma }_{k}}(1 - \rho _k^2)}}} \right),
\end{equation}
where $I_0$ is the zero-order modified Bessel function of the first kind.
We further denote by ${\bar {\bar \gamma} _{k,i}}$ the median of the instantaneous SNR ${\gamma _{k,i}}$, for which the following equation holds:
 \begin{equation}
\label{eq:Relationship_half_probability}
\int\nolimits_0^{{\bar {\bar \gamma} _{k,i}}}\!\! \!\! \!\! {\mathop\mathbb{P}\left[ \gamma _{k,i} | \hat{\gamma} _{k,i}\right]d} {\gamma _{k,i}} = \int\nolimits_{{ \bar {\bar \gamma}_{k,i}}}^{ + \infty }\!\! \!\! \!\!  {\mathop\mathbb{P}\left[ \gamma _{k,i} | \hat{\gamma} _{k,i}\right]d} {\gamma _{k,i}} =0.5 \ .
\end{equation}
Due to~\eqref{eq:Relationship_original} the median of the distribution of the instantaneous channel ${h_{k,i}}$ is ${\rho _k}{\hat h_{k,i}}$, thus we have ${\bar {\bar \gamma} _{k,i}}  \approx \rho^2 _k {  \gamma} _{k,i}$.

\subsection{Finite Blocklength Error Model under Perfect CSI} 
\label{sec:FBL_one_hop}
For  the real additive white Gaussian noise (AWGN) channel~[8, Theorem 54] derives an accurate approximation of the
coding rate of a single-hop transmission system. With blocklength $m$, block error probability $\varepsilon$ and SNR $\gamma$, the coding rate (in bits per channel use) is given by
 $r \approx \frac{1}{2}{\log _2}\left( {1 + \gamma } \right) - \sqrt {\frac{V_{\text{real}}}{m}} {Q^{ - 1}}\left( \varepsilon  \right)$, where~$Q^{ - 1}(\cdot)$ is the inverse of the Q-function given by $Q\left( w \right) = {\rm{ }}\int_w^\infty  {\frac{1}{{\sqrt {2\pi } }}} e^{ - t^2 /2} dt$. In addition, $V_{\text{real}}$ is the  channel dispersion of a real Gaussian channel which is given by $V_{\text{real}} = \frac{\gamma }{2}\frac{{\gamma  + 2}}{{{{\left( {1 + \gamma } \right)}^2}}}{\left( {{{\log }_2}e} \right)^2}$.
 
{{Under a quasi-static fading channel model, each channel state is assumed to be static during a frame, i.e., in each frame a quasi-static fading channel with fading coefficient $h$ can be viewed as an AWGN channel with channel gain $|h|^2$.}} 
Therefore, the above result of the real AWGN channel has been extended to a complex quasi-static fading channel model\cite{Yang_2014,Gursoy_2013,Gursoy_2013_2,Peng_2011,Makki_2014}: For a received SNR~$\gamma$, the coding rate of a frame (in bits per channel use) is given by: 
\begin{equation}
\label{eq:single_link_data_rate}
r =  {\mathcal{R}}({\gamma},\varepsilon ,m)  \approx {\mathcal{C}}({\gamma}) - \sqrt {\frac{{{V_{{\rm{comp}}}}}}{m}} {Q^{ - 1}}\left( \varepsilon  \right)\mathrm{,}
\end{equation}
where ${\mathcal{C}}\left( \gamma \right)$ is the  Shannon capacity function of a complex channel with received SNR~$\gamma$ : $ {\mathcal{C}}(\gamma) \!=\! {\log _2}\left( 1 + \gamma \right)$.  
{{In addition,  the  channel dispersion of a complex Gaussian channel is  twice the one of a real Gaussian channel: ${V_{{\rm{comp}}}} = 2{V_{{\rm{real}}}} = \gamma \frac{{\gamma  + 2}}{{{{\left( {1 + \gamma } \right)}^2}}}{\left( {{{\log }_2}e} \right)^2} = \left( {1 - \frac{1}{{{{\left( {1 + \gamma } \right)}^2}}}} \right){\left( {{{\log }_2}e} \right)^2} $.}}

Then,  for a single-hop transmission under a quasi-static fading channel,  
with blocklength $m$ and coding rate $r$, 
 the decoding (block) error probability at the receiver is given by:  
\begin{equation}
\label{eq:single_link_error_pro}
\varepsilon  = {\mathcal{P}}(\gamma,r,m)  \approx Q\left( {\frac{{{\mathcal{C}}(\gamma) - r}}{{\sqrt {{V_{{\rm{comp}}}}{\rm{/}}m} }}} \right).
\end{equation}
{{Considering the channel fading, the expected/average error probability 
 is given by~\cite{Yang_2014}:
\begin{equation}
\label{eq:expectedsingle_link_error}
\mathop {\mathop{\mathbb{E}}\nolimits} \limits_{\gamma} \left[ \varepsilon  \right] = \mathop {\mathop{\mathbb{E}}\nolimits} \limits_{\gamma} \left[ {{\mathcal{P}}(\gamma,r,m)} \right]  \approx \mathop {\mathop{\mathbb{E}}\nolimits} \limits_{\gamma} \left[ {Q\left( {\frac{{{\mathcal{C}}(\gamma) - r}}{{\sqrt {{V_{{\rm{comp}}}}{\rm{/}}m} }}} \right)} \right].
\end{equation}

In the remainder of the paper, we investigate the considered relaying system in the finite blocklength regime by applying the above approximations. As these approximations have been shown to
be accurate for a sufficiently large blocklength $m$~\cite{Verdu_2010}, for simplicity we will assume them to hold in equality in our analysis 
 and numerical evaluation conditioned on the assumption of a sufficiently large value of $m$ at each hop. 

\section{Maximizing the FBL Throughput under Reliability Constraints}
\label{sec:Physical_layer_P}
As discussed in the previous section, the source has outdated channel state information that it can rely on for scheduling the data transmission in the relay system. 
In this section, we address thus the problem of how to optimally schedule the coding rate based on the inaccurate outdated CSI such that the throughput of the relay system is maximized. 
We restrict this scheduling problem to a reliability constraint such that for each data transmission a target error probability $\varepsilon_{\rm th}$ must be met.
Such a scheduling problem is justified by current discussions around industrial wireless communication systems, where small payload packets need to be transmitted within a bounded time interval while keeping a (stochastic) reliability guarantee. 
In the following, we first develop a throughput model of the relaying system with respect to the finite blocklength assumption, building on Section~\ref{sec:FBL_one_hop}.
Subsequently, the mathematical statement of the optimization problem is provided. 
We then turn to the solution, providing both an optimal solution as well as a low-complexity heuristic.

 \subsection{FBL Throughput Model for Relay Systems}

Assuming $r_i$ is the scheduled coding rate\footnote{Note that only a single coding rate is scheduled for both links per frame.} for frame $i$ with instantaneous SNRs $\gamma_{1,i}$ and $\gamma_{2,i}$, the overall error probability of the relaying system during frame $i$ is: 
\begin{equation}
\label{eq:Overall_error_pro}
\begin{split}
{\varepsilon _{{\rm{R}},i}}({r_i}) &= 1 - (1 - {\varepsilon _{1,i}})(1 - {\varepsilon _{2,i}})\\
& =  {\varepsilon _{1,i}} + {\varepsilon _{2,i}} - {\varepsilon _{1,i}}{\varepsilon _{2,i}}\mathrm{,}
\end{split} 
\end{equation} 
where $ \varepsilon _{k,i}= {\mathop{\mathcal{P}}} (\gamma_{k,i},r_i,m), i=1 ,2$.
Based on~\eqref{eq:Overall_error_pro}, we immediately have the expected overall error probability conditioned on the outdated CSI $\hat{\gamma}$.
It is the expected value of~\eqref{eq:Overall_error_pro} over the conditioned channel fading distribution: 
\begin{equation}
\label{eq:Overall_error_pro_overchannelfading}
{\bar \varepsilon _{{\rm{R}},i}}({r_i}) = {\bar \varepsilon _{1,i}} + {\bar \varepsilon _{2,i}} - {\bar \varepsilon _{1,i}}{\bar \varepsilon _{2,i}} \mathrm{.}
\end{equation}
In~\eqref{eq:Overall_error_pro_overchannelfading}, $\bar \varepsilon _{k,i} $, $k=1,2$ are the expected error probabilities of either the backhaul link or the relaying link.  
\begin{table*}[t]
\begin{equation} 
\label{eq:single_expected_error_prob}
\begin{split}
\! \! \! \! {{\bar \varepsilon }_{k\!,i}} &\!=\! \int\limits_0^{ + \infty } \!{\mathbb{P}\left[ \gamma _{k,i} | \hat{\gamma} _{k,i}\right]{\mathcal{P}}({\gamma _{\!k\!,i}},\!{r_i},\!m)d} {\gamma _{\!k\!,i}}
 \!=\! \int\limits_0^{ + \infty }\!  {\frac{{{I_0}( {\frac{{2{\rho _k}\sqrt {{\gamma _{\!k\!,i}}{{\hat \gamma }_{k,i}}} }}{{{{\bar \gamma }_{k\!}}(1 - \rho _k^2)}}} ){\exp{( - \frac{{{\gamma _{\!k\!,i}} + \rho _k^2{{\hat \gamma }_{\!k\!,i}}}}{{{{\bar \gamma }_{k\!}}(1 - \rho _k^2)}})}}}}{{{{\bar \gamma }_{k\!}}(1 - \rho _k^2)}}Q( {\frac{{{\mathcal{C}}({\gamma _{k\!,i}}) - {r_i}}}{{\sqrt {\frac{1}{m}\left( {1 - {2^{ - 2{\mathcal{C}}({\gamma _{k\!,i}})}}} \right)} {{\log }_2}e}}} )} d{\gamma _{k\!,i}}\\
 &\!=\! \frac{1}{{\!\sqrt {\!2\pi } }}\!\! \!   \int\limits_0^{ + \infty }\!  {\!\int\limits_{\alpha({\gamma _{\!k\!,i}},{r_i})}^\infty \! \! \! \! {\frac{{{I_0}( {\frac{{2{\rho _k}\sqrt {{\gamma _{\!k\!,i}}{{\hat \gamma }_{k,i}}} }}{{{{\bar \gamma }_{k\!}}(1 - \rho _k^2)}}})}}{{{{\bar \gamma }_{k\!}}(1 - \rho _k^2)}}} {e^{ - \frac{{{\gamma _{\!k\!,i}} + \rho _k^2{{\hat \gamma }_{\!k\!,i}}}}{{{{\bar \gamma }_{k\!}}(1 - \rho _k^2)}} - \frac{{{t^2}}}{2}}}dtd} {\gamma _{\!k\!,i}}
\!  \mathop  = \limits^{x =\! \sqrt {\!\frac{{2{\gamma _{k,i}}}}{{{{\bar \gamma }_{\!k\!, }}\!(\!1 \!-\! \rho _k^2\!)\!}}} } \! \!\frac{1}{{\sqrt {2\pi } }}\!\int\limits_0^{ + \infty } \! \!{\int\limits_{\alpha(x,{r_i})}^\infty   \! \!{x{I_0}( {\frac{{x{\rho _k}\sqrt {2{{\hat \gamma }_{k,i}}} }}{{\sqrt {{{\bar \gamma }_{k}}(1 - \rho _k^2)} }}} )} {e^{ \!- \frac{{{x^2}}}{2} - \frac{{\rho _k^2{{\hat \gamma }_{k,i}}}}{{{{\bar \gamma }_{k}}(1 - \rho _k^2)}} - \frac{{{t^2}}}{2}}}\!dtd} x   \\
&\!=\!
\frac{1}{{\sqrt {2\pi } }}\int\limits_0^{ + \infty }\! {\int\limits_{\alpha(x,{r_i})}^\infty  \!\!{x{I_0}\left( { x{\rho _k}{x_{k,i}}} \right)} {e^{ - \frac{{{x^2} + \rho _k^2{x_{k,i}} + {t^2}}}{2}}}dtdx} . 
\end{split}
 \end{equation}
\vspace*{-2pt}
\hrule
\vspace*{-5pt}
\end{table*}
%
Then,  by averaging $\varepsilon _{k,i} $ over  the  conditional PDF in~\eqref{eq:Bessel_probability}, $\bar \varepsilon _{k,i} $, $k=1,2$   is given by~\eqref{eq:single_expected_error_prob}, where $\alpha(x,{r_i}) \!=\! \frac{{{\mathcal{C}}({x^2}{{\hat \gamma }_{k,i}}(1 - \rho _k^2)/2) - {r_i}}}{{\sqrt {\frac{1}{m}( {1 - {2^{ - 2{\mathcal{C}}({x^2}{{\hat \gamma }_{k,i}}(1 - \rho _k^2)/2)}}} )} {{\log }_2}e}}$ and ${x_{k,i}} = \sqrt {\frac{{2{{\hat \gamma }_{k,i}}}}{{{{\bar \gamma }_k}(1 - \rho _k^2)}}}$.

Notice that for the relaying system considered,  the (source-to-destination) equivalent coding rate during each frame $i$ is actually $r_i/2$. 
Therefore, the expected FBL throughput of relaying during frame $i$, i.e., the expected effectively transmitted information (number of correctly received bits at the destination) per channel use,  is given by:  
\begin{equation}
\label{eq:BL_capacity_of_a_frame}
 {\mu}_{{\rm{FBL}},i} = \mathcal{C}_{{\rm{FBL}}} (r_i) = r_i{(1-\bar \varepsilon _{{\rm{R}},i} (r_i)  )} /2 \mathrm{.}  
\end{equation}
The above ${\mu}_{{\rm{FBL}},i}$ is the expected FBL throughput of relaying for an upcoming frame $i$ based on a scheduled coding rate.  
By marginalizing over all possible channel states for both links, we finally end up with the average FBL throughput of relaying:
${{\mu}_{{\rm{FBL}}}} \!=\! \mathop {\mathbb{E}}\limits_{i=1,...,+\infty } \left[ {{\mathcal{C}_{{\rm{FBL}},i}}}  (r_i)  \right] \!=\! \mathop {\mathbb{E}}\limits_{\gamma _{1,i},\gamma _{2,i}} \left[ {{\mathcal{C}_{{\rm{FBL}},i}}}  (r_i)  \right]$.
Note in particular that this throughput depends on the scheduled coding rate $r_i$ which itself can be based on the information at hand of the source, i.e. the outdated SNRs $\hat \gamma _{1,i}$ and $\hat \gamma _{2,i}$. 
 
\subsection{Optimal Scheduling}
\label{sec:OptimalScheduling}
Recall that we are interested in scenarios  with reliability constraints,  i.e.,  the (expected/average) error probability of each link should be lower than a threshold  $\varepsilon_{\rm th}$ of practical interest, e.g., $\varepsilon_{\rm th} \ll 0.5$. 
If the source schedules the coding rate directly based on the outdated CSI, it is likely that the instantaneous SNR is lower. 
This can introduce a significant source of block errors while generally leading to a higher coding rate in case of successful transmissions, i.e. we face a typical trade-off.
To study this trade-off, we introduce weights, i.e., SNR back-offs, to let the source choose a relatively lower coding rate  obtained by scaling the outdated SNR.  
Denote these weights for frame $i$  for the backhaul link by $\eta_{1,i}$  and for the relaying link by $\eta_{2,i}$, where $0< \eta_{k,i}, k=1,2$. 
Recall that the performance of the two-hop relaying system is subject to the bottleneck link which can be either the backhaul or the relaying link. 
Thus, for a given selection of the weights $\eta$ the coding rate~$r_i$ of frame $i$ is determined based on the bottleneck link: $r_i =  {\mathcal{R}}(\min \{\eta_{1,i} {{\hat \gamma}}_{1,i}, \eta_{2,i} {{\hat \gamma}}_{2,i} \},\varepsilon_{\rm th} ,m)$. 

Our aim is to determine - per frame - the optimal scheduling weights (of the backhaul link and the relaying link) for coding rate scheduling which maximizes the average FBL throughput while guaranteeing the reliability of transmissions. 
Therefore, the optimization problem actually equals to maximize the expected FBL throughput per frame by solving the following optimization problem:
 \begin{equation}
\begin{split}
\label{eq: problem_general1}
 \mathop {\max }\limits_{\eta_{1,i}, \eta_{2,i}}  \     \      &{{\mu}_{{\rm{FBL}}}}  \\ 
s.t.: \    \ &  \bar \varepsilon_ {k,i}\!\le\! \varepsilon_{\rm th}, k=1,2; i = 1,...,+\infty. 
\end{split}
\end{equation} 
For this optimization problem, note that the space of feasible solutions for the scheduling weights is restricted in the following way:
We are interested in reliable transmission, i.e. we restrict the transmission to the reliability constraint $\epsilon_{\rm th} \ll 0.5$. Then, according to~\eqref{eq:Relationship_half_probability} we have   
$\textstyle{{\bar \varepsilon }_{k,i}} \!\!=\! \mathop {\mathbb{E}}\limits_{{\gamma _{k,i}}|{{\hat \gamma }_{k,i}}} \left[ {{\varepsilon _{k,i}}} \right] \le 0.5\Leftrightarrow \mathbb{P}  \{ {\mathop {\mathbb{E}}\limits_{{\gamma _{k,i}}|{{\hat \gamma }_{k,i}}}\!\! \left[ {\mathcal C({\gamma _{k,i}})} \right] \ge {r_i}}  \} \ge 0.5   \Leftrightarrow \mathbb{P} \left\{ {{\gamma _{k,i}} \ge {\eta _{k,i}}{{\hat \gamma }_{k,i}}} \right\} \ge 0.5 \Leftrightarrow {\eta _{k,i}} \le \rho _k^2$, 
which results thus in the space $\eta_{k,i} \in \left[0, \rho _k^2\right]$. 

Under this constraint, the following proposition can be shown with respect to the scheduling of the weights for the considered relay system:
\newtheorem{theorem}{Proposition}
\begin{theorem}
\label{th:monn-dec} 
{For a relay network operating on outdated CSI, if  the coding rate for frame $i$ is scheduled according to  $r_i =  {\mathcal{R}}(\min \{\eta_{1,i} {{\hat \gamma}}_{1,i}, \eta_{2,i} {{\hat \gamma}}_{2,i} \},\varepsilon_{\rm th}, m)$,  $\eta_{k,i} \in (0, {{\rho^2 _k}}], k =1 ,2$,   the expected  FBL throughput  of the upcoming frame~$i$, $ {\mu}_{{\rm{FBL}},i}  = \mathcal{C}_{{\rm{FBL}}} (r_i)$,  is concave in the coding rate~$r_i$.}
\end{theorem}

\begin{proof} {See Appendix A.}
\end{proof}

Recall that the coding rate is chosen by the source based on $\min \{\eta_{1,i} {{\hat \gamma}}_{1,i}, \eta_{2,i} {{\hat \gamma}}_{2,i} \}$. 
Due to~\eqref{eq:single_link_data_rate}, the coding rate is strictly increasing in $\min \{\eta_{1,i} {{\hat \gamma}}_{1,i}, \eta_{2,i} {{\hat \gamma}}_{2,i} \}$ and therefore increasing in $\eta_{1,i}$ or  $\eta_{2,i}$. 
In combination with {Proposition~1}, we thus obtain an important corollary regarding the optimal scheduling of the system:

\newtheorem{theorem1}{Corollary}
\begin{theorem1}
{For a relay network operating on outdated CSI, if  the coding rate for frame $i$ is scheduled according to  $r_i =  {\mathcal{R}}(\min \{\eta_{1,i} {{\hat \gamma}}_{1,i}, \eta_{2,i} {{\hat \gamma}}_{2,i} \},\varepsilon_{\rm th}, m)$,  $\eta_{k,i} \in (0, {{\rho^2 _k}}], k =1 ,2$,  
 $\mathcal{C}_{{\rm{FBL}},i}$ the expected  FBL throughput of frame $i$ is quasi-concave in $\eta_{1,i}$  in the region $(0, {{\rho^2 _1}}]$ and quasi-concave in  $\eta_{2,i}$ in the region $(0, {{\rho^2 _2}}]$.}
\end{theorem1}
\begin{proof} {See Appendix B.}
\end{proof}

According to Corollary 1, the expected FBL throughput of frame $i$ ${{C_{{\rm{FBL}},i}}}$ can be optimized by applying quasi-convex optimization techniques, e.g., backtracking line search, to obtain the optimal weights for determining the coding rate.  
Nevertheless, this can be computationally heavy, as this optimization step needs to be conducted prior to each data transmission. 
Note in this context that   the smaller the reliability requirement is, the smaller is also the search space of the scheduling weights, making it more likely that for a given instance the optimal solution is on the boundary of the feasible set. Still, in order to reach the optimal system performance, some computations need to be executed prior to each frame.

\subsection{Constant Weight Heuristic}
\label{sec:Fixed}
To further reduce the computational complexity, in this section we consider scheduling schemes where the weight is not adapted per frame.
Once the scheduling weights are determined at system initialization (depending on the average SNR and the correlation coefficients)  they remain constant during all frames.
We are interested in determining the constant heuristic with the best performance.

Denote these constant weights by $ \eta_1$ and $\eta_2$ for the backhaul and relaying link.  
Then, the coding rate for frame $i$ under the constant weight scheme is subject to the instantaneous SNR and the constant weights. 
As a result, obviously the coding rate is not constant over different frames. 
In particular, the coding rate~$r_i$ of frame $i$ is obtained by: $r_i \!\! =\! \!  {\mathcal{R}}(\min \{\eta_{1} {{\hat \gamma}}_{1,i}, \eta_{2} {{\hat \gamma}}_{2,i} \},\varepsilon_{\rm th}, m)$. 
According to~\eqref{eq:single_link_data_rate}, the coding rate $r_i $ is strictly increasing in $\min \{\eta_{1} {{\hat \gamma}}_{1,i}, \eta_{2} {{\hat \gamma}}_{2,i} \}$ and therefore
monotonically increasing in $\eta_1$ and $\eta_2$. 
Thus, under this constant weight scheme, the average FBL throughput can be determined by:
\begin{IEEEeqnarray}{LLL}
\label{eq: FBL throughput_over_fading}
\begin{split}
\!\!&\!\!{{\mu}_{{\rm{FBL}}}} (\eta_1,\eta_2)  \!=\! \mathop {\mathbb{E}}\limits_{{r_i}} \left[ {{\mathcal{C}_{{\rm{FBL}}}}({r_i})} \right]  \\
\!= & \!\int_0^\infty \!\! \! {\int_0^\infty   \!\!\!\! {{\mathcal C_{{\rm{FBL}}}}\left( {{\mathcal{R}}(\min \{ {\eta _{\rm{1}}}{{\hat \gamma }_{\rm{1}}},{\eta _{\rm{2}}}{{\hat \gamma }_{\rm{2}}}\} ,\varepsilon_{\rm th}, m)} \right)} } {e^{ \!-\! \frac{{{{\hat \gamma }_1}}}{{{{\bar \gamma }_1}}} \!-\! \frac{{{{\hat \gamma }_2}}}{{{{\bar \gamma }_2}}}}}d\frac{{{{\hat \gamma }_{\rm{1}}}}}{{\bar \gamma }}d\frac{{{{\hat \gamma }_2}}}{{\bar \gamma }}\\
  \!=&  \frac{1}{{{{\bar \gamma }_1}\!{{\bar \gamma }_2}}}\!\int_0^\infty  \!{\int_{\frac{{{\eta _1}{{\hat \gamma }_{\rm{1}}}}}{{{\eta _2}{{\bar \gamma }_2}}}}^\infty  \! {{\mathcal{C}_{{\rm{FBL}}}}\left( {{\mathcal{R}}({\eta _{\rm{1}}}\!{{\hat \gamma }_{\rm{1}}},\varepsilon_{\rm th}, m)} \right)} } {e^{ \!- \frac{{{{\hat \gamma }_1}}}{{{{\bar \gamma }_1}}} \!-\! \frac{{{{\hat \gamma }_2}}}{{{{\bar \gamma }_2}}}}}d{{\hat \gamma }_2}d{{\hat \gamma }_1}\\
  & \!+\! \frac{1}{{{{\bar \gamma }_1}\!{{\bar \gamma }_2}}}\!\int_0^\infty  \!{\int_{\frac{{{\eta _2}{{\hat \gamma }_2}}}{{{\eta _1}{{\bar \gamma }_1}}}}^\infty\!  {{\mathcal{C}_{{\rm{FBL}}}}\!\left( {{\mathcal{R}}({\eta _2}{{\hat \gamma }_2},\varepsilon_{\rm th}, m)} \right)} } {e^{ \!- \frac{{{{\hat \gamma }_1}}}{{{{\bar \gamma }_1}}} \!- \frac{{{{\hat \gamma }_2}}}{{{{\bar \gamma }_2}}}}}d{{\hat \gamma }_1}d{{\hat \gamma }_2}.
\end{split}
 \end{IEEEeqnarray}
Under the best constant heuristic, the aim is to maximize the average FBL throughput while the constraint is to guarantee the average error probability over time~\footnote{From a statistical point of view, guaranteeing the expected error probability per frame leads to the same results as guaranteeing the average error probability over time}.  
Then, the resulting optimization problem is given by: 
\begin{equation}
\begin{split}
\label{eq: problem_constant}
\mathop {\max }\limits_{\eta_{1},\eta_{2} } \quad  &{\mu}_{{\rm{FBL}}} (\eta_1,\eta_2). \\
s.t: \quad & \!\!\!\mathop \mathbb{E} \limits_{{\gamma _{k,i}}} [  \bar \varepsilon_ {k,i}]\!\le\! \varepsilon_{\rm th}, k=1,2 .  
\end{split}
\end{equation} 
As  we assume $\varepsilon_{\rm th} \ll 0.5$, we have
$\!\!\mathop \mathbb{E} \limits_{{\gamma _{k\!,i}}}\! [\bar \varepsilon_ {k\!,i}]\! \!=\! \!\mathop {\mathbb{E}}\limits_{{{\hat \gamma }_{k\!,i}}}\! [{\mathop {\mathbb{E}}\limits_{{\gamma _{k\!,i}}\!|{{\hat \gamma }_{k\!,i}}}\!\! [ {{\varepsilon _{k\!,i}}} ]} ]\! \le\! 0.5 \!\!\Leftrightarrow$ $   \! \! \! \mathop {\mathbb{E}}\limits_{{{\hat \gamma }_{k\!,i}}} \![ {\mathop {\mathbb{E}}\limits_{{\gamma _{k\!,i}}\!|{{\hat \gamma }_{k\!,i}}}\!\!\! [ {Q ( {\frac{{\mathcal C({\gamma _{k\!,i}}) - {r_i}}}{{\sqrt {{V_{{\rm{comp}}}}{\rm{/}}m} }}}  )} ]} ]  \!\!\!\le\!\! 0.5 \!\!\Leftrightarrow \!\!\mathop {\rm{E}}\limits_{{{\hat \gamma }_{k\!,i}}}\! [ {\mathbb{P}  \{\!\!\! \!{\mathop {\mathbb{E}}\limits_{{\gamma _{k\!,i}}\!|{{\hat \gamma }_{k\!,i}}}\!\!  \! \![ {\mathcal C({\gamma _{k\!,i}})} ]   \!\!\ge  \! \!{r_i}}  \}} ]\!\! \ge\!\! 0.5 \Leftrightarrow \mathop {\mathbb{E}}\limits_{{{\hat \gamma }_{k\!,i}}} [ {\mathbb{P} \left\{ {{\gamma _{k\!,i}} \ge {\eta _{k\!}}{{\hat \gamma }_{k\!,i}}} \right\}} ] \ge 0.5 \Leftrightarrow {\eta _{k\!}} \le \rho _k^2$. Therefore, the feasible set of   $\eta_{k,i}$ is $(0, \rho _k^2]$ under the case  $\varepsilon_{\rm th}=0.5$ and  covers a subset of $(0, \rho _k^2]$ when $\varepsilon_{\rm th} \ll 0.5$. 

Denote  $\eta^*_1$ and~$\eta^*_2$ as the solution to the above optimization problem, i.e. they are the optimal, constant scheduling weights. 
We then have the following proposition:

\begin{theorem}
{Considering a relay network operating on outdated CSI with constant scheduling weights, if the coding rate of each frame  $i$  is scheduled according to $r_i =  {\mathcal{R}}(\min \{\eta_{1} {{\hat \gamma}}_{1,i}, \eta_{2} {{\hat \gamma}}_{2,i} \},\varepsilon_{\rm th}, m)$, 
 then the average FBL throughput ${C}_{{\rm{FBL}}}$ is quasi-concave in $\eta_{1}$  in the region $(0, {{\rho^2 _1}}]$ and quasi-concave in  $\eta_{2}$ in the region $(0, {{\rho^2 _2}}]$}.
\end{theorem}
\begin{proof} {See Appendix C.}
\end{proof}

According to Proposition 2, \eqref{eq:  problem_constant} can be efficiently solved by applying quasi-convex optimization techniques.  
For a relay system with a certain set of average SNRs and correlation coefficients as well as   a given reliability constraint, we obtain a unique pair of fixed scheduling weights. 
Note that these fixed weights are then strictly applied per frame, leading to a varying coding rate that maximizes the long-term average FBL throughput (under the assumption of using fixed weights).
This reduces drastically the computational complexity, but leads to an inferior system performance in comparison to the optimal scheduling scheme with adaptive scheduling weights, i.e. the optimal solution presented in Section~\ref{sec:OptimalScheduling}.
\section{Numerical Evaluation and Discussion}
\label{sec:simulation}

In this section, we present some numerical results regrading the throughput maximization in relay systems.
We consider in particular two issues: Initially, we study several aspects of the quasi-convexity of the FBL throughput with respect to the scheduling weights. In particular, we are interested in the sharpness of the optimum. This investigation is important for practical system design, as it clarifies the potential cost of non-optimal weight selection.
After clarifying these issues, we move to a more general performance investigation. Here, we are especially interested in the performance comparison between the optimal scheduling scheme (with changing scheduling weights per frame) and the low-complexity best constant heuristic. 

From a methodological point of view, all following numerical results are based on simulations.  
We consider a basic scenario for these simulations with the following parameterization: 
We assume an urban outdoor scenario where the distances of the backhaul and relaying link are both set to 100~$\rm{m}$. For channel propagation, we utilize the well-known COST~\cite{Molisch_2011} model (which is a commonly-used model for urban scenarios) for calculating the path loss. The center frequency is set to to 2~$\rm{GHz}$ while the transmit power $p_{\rm{tx}}$ is selected to 35~$\rm{dBm}$ (we vary the transmit power in  Figures~\ref{BL_capa_threshold} and \ref{Dynamic_fixed_blt}) considering a noise power of -90~$\rm{dBm}$,  respectively.  Lastly, the blocklength at each hop of relaying is set to $m = 300$ symbols\footnote{From [8, Figure 2] it is known that the relative difference of the approximate and the exact achievable rates is less than 2\% for cases with $m \ge$ 100.}. 
Recall that the channel correlation coefficients $\rho^2_1$ and $\rho^2_1$ of the backhaul and relaying links are subject to the settings of $n$, $n+m$ and the Doppler frequency. In particular, $\rho^2_1 \ge \rho^2_2$ as $n \le n+m$, i.e., the CSI of the relaying phase is more delayed. 
In the simulation, we don't set a fixed value for either the length of initialization phase $n$ or the Doppler frequency. Instead, we consider different setups of  $\rho^2_1$ and $\rho^2_1$, which corresponds to different settings of $n$ and the Doppler frequency as $m$ is fixed, while $\rho^2_1 \ge \rho^2_2$ holds  for all setups.

\subsection{Quasi-Convexity of the FBL Throughput}
\label{sec:simulation1}
In this subsection, we consider numerical results regarding the quasi-convexity of the average FBL throughput. 
We first study the relationship between the expected FBL throughput of an upcoming frame and the choice of scheduling weights (based on the corresponding outdated CSI) in case of the optimal scheduling scheme that adapts the weights per frame. 
In order to do so, we fix the outdated CSI and generate realizations of the corresponding instantaneous channel states. 
Then, we study the expected FBL throughput (the expectation/average over all realizations) by varying the scheduling weights. 
The results are shown in Figure~\ref{BL_capa_perframeweights}. 
\begin{figure}[!h]
\centering{\includegraphics[width=100mm, trim=0 25 0 26]{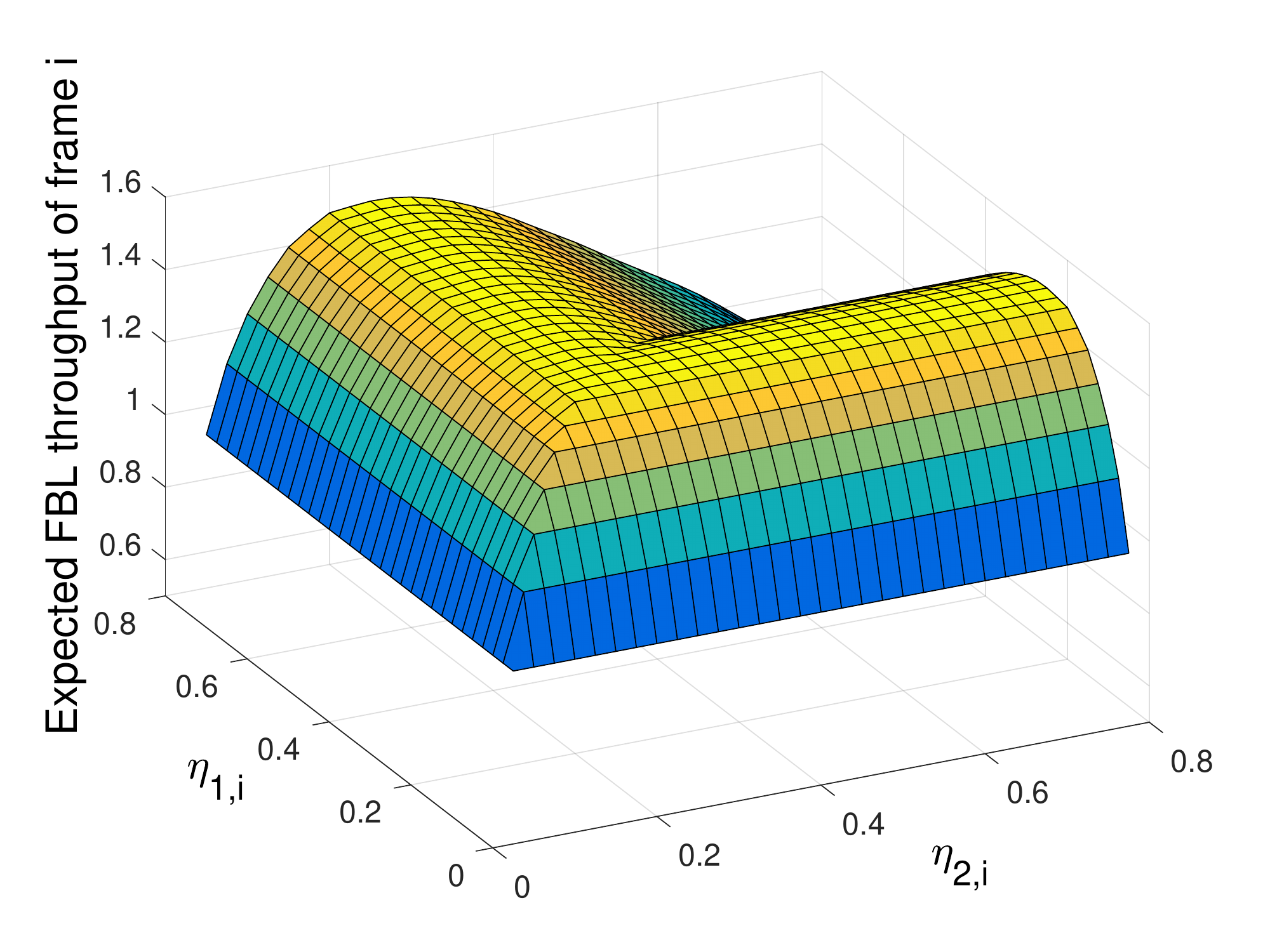}}
\caption{Expected FBL throughput [bit/ch.use] of an upcoming frame $i$  vs. the choice of scheduling weights. In the figure, we set  $\rho^2_1=0.7$ and $\rho^2_2=0.5$.}
\label{BL_capa_perframeweights}
\end{figure}
First of all, the figure illustrates that the FBL throughput per frame is quasi-concave in the scheduling weights~$\eta_{1,i}$ and~$\eta_{2,i}$.  
Hence, by choosing appropriate values for~$\eta_{1,i}$ and~$\eta_{2,i}$ the FBL throughput can be optimized. 
Secondly, the figure also shows that the expected FBL throughput of the upcoming frame is actually subject to both scheduling weights $\eta_{1,i}$  and  $\eta_{2,i}$ in general. 
However, if $\eta_{1,i}$ is chosen optimally, then the choice of $\eta_{2,i}$ can be arbitrarily large (but not arbitrarily small). 
This stems essentially from the fact how the scheduling weights influence the bottleneck link. The case "FBL throughput being only influenced by  $\eta_{1,i}$" corresponds to the situation where the bottleneck link (for determining the coding rate) is the backhaul link. At the same time,~as long as $\eta_{2,i}$ (the scheduling weight of the relaying link) is not set to a very small value, the impact on the SNR of the backhaul link is considerably small and therefore does not influence the coding rate. In other words, there is no impact of a link's scheduling weight on the FBL throughput of the upcoming frame if this link is not the bottleneck link. Obviously, reducing the scheduling weight of a link likely makes this link become the bottleneck link eventually. 
As a consequence of this dependence between $\eta_{1,i}$  and $\eta_{2,i}$, we observe that there are multiple solutions maximizing the FBL throughput surface for the considered channel setting.

\begin{figure}[!t]
\centering{\includegraphics[width=100mm, trim=0 10 0 25]{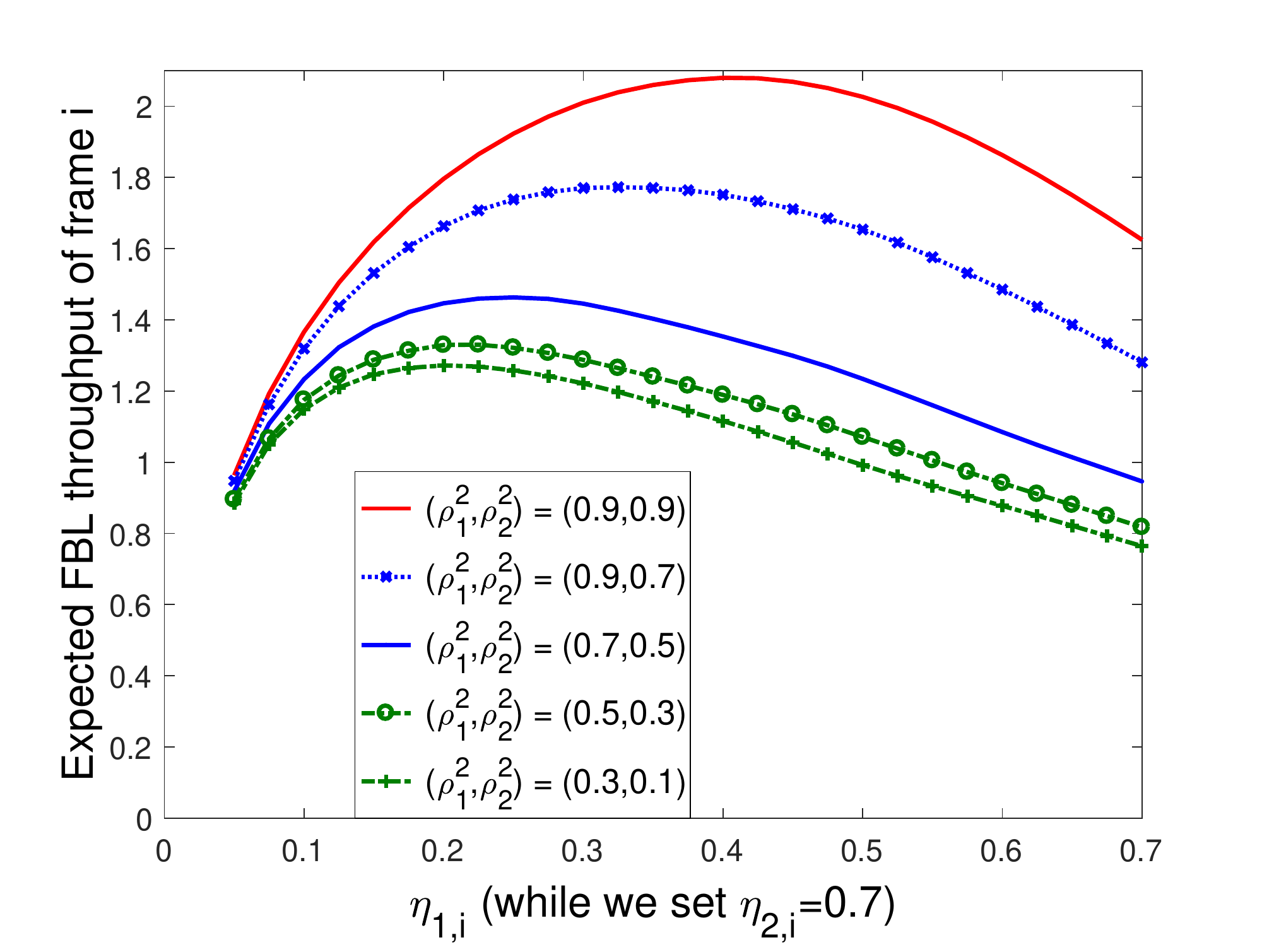}}
\caption{Expected FBL throughput [bit/ch.use] of an upcoming frame  $i$ with different channel correlation coefficients. In the figure, we vary the scheduling weight $\eta_{1,i}$ while setting $\eta_{2,i}=0.7$.}
\label{BL_capa_perframeweights2}
\end{figure}
We next study the quasi-convexity of the optimal scheduling for scenarios with different channel correlation setups. 
The results are shown in Figure~\ref{BL_capa_perframeweights2} where we fix $\eta_{2,i}$ to 0.7, i.e., make the backhaul link the bottleneck and vary $\eta_{1,i}$. 
The figure reveals that a stronger channel correlation results in a higher optimal FBL throughput. 
More importantly, this higher maximum is achieved by a bigger scheduling weight. 
In other words, a strong channel correlation allows us to set the scheduling weight more aggressively, leading to a higher coding rate and a higher FBL throughput.

We now turn to the constant heuristic, where the scheduling weight is determined once for the given system and then left constant for each frame.   
In Figure~\ref{BL_capa_fixedweights}, we show the relationship between the average FBL throughput $\mathcal{C}_{\rm BL}$ and the constant scheduling weights  $\eta_1$ and $\eta_2$ while generating many different outdated channel instances and the corresponding instantaneous channel state realizations. 
\begin{figure}[!t]
\centering{\includegraphics[width=100mm, trim=0 18 0 32]{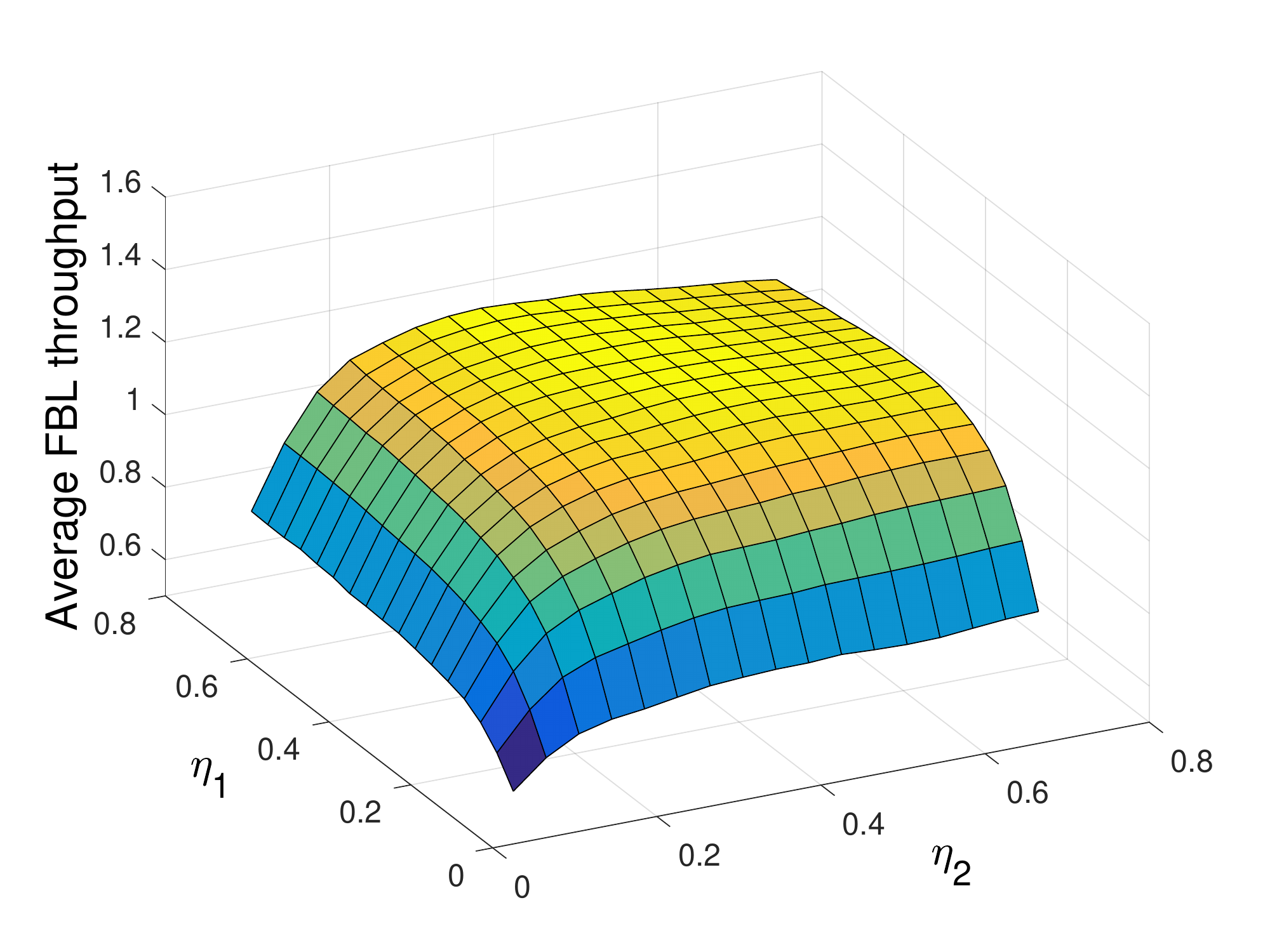}}
\caption{Average FBL throughput [bit/ch.use] vs. choice of constant scheduling weights for a relay system with parameters $\rho^2_1=0.7$ and $\rho^2_2=0.5$.}
\label{BL_capa_fixedweights}
\end{figure}
\begin{figure}[!t]
\centering{\includegraphics[width=100mm, trim=0 13 0 25]{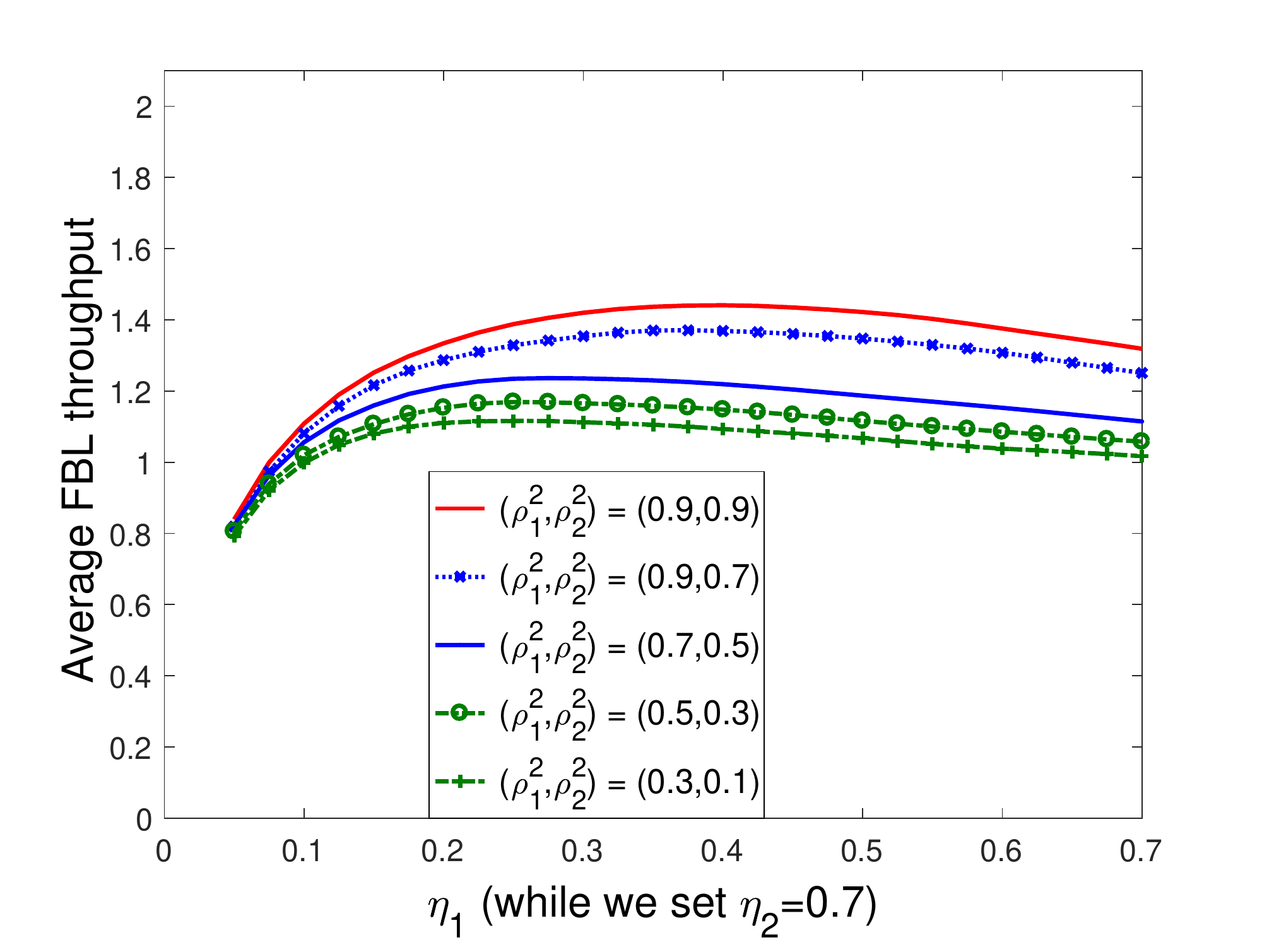}}
\caption{Average FBL throughput [bit/ch.use] vs. choice of constant scheduling weights for relay systems with different channel correlation coefficients. In the figure, we vary the fixed scheduling weight $\eta_{1}$ while setting $\eta_{2}=0.7$.}
\label{BL_capa_fixedweights2}
\end{figure}
Firstly, the figure confirms again our analytical insight (Proposition~2), i.e.   ${C_{{\rm{FBL}}}} $ is  quasi-concave in $\eta_1$ or $\eta_2$. 
In addition, we observe that a near-optimal FBL throughput is achieved for a large set of different scheduling weights, e.g., a small error of the optimal solution does not change the average FBL throughput too much.
Hence, the FBL throughput in the case of the constant scheduling weights is somewhat robust to an erroneous choice of the weights. 
Similar to the optimal scheduling, we further study the average FBL throughput of the best constant heuristic  with different channel correlation coefficients. 
The results are provided in Figure~\ref{BL_capa_fixedweights2}. 
It is shown that under the best constant heuristic scheme a strong channel correlation also introduces a higher FBL throughput attained for larger scheduling weights.
Nevertheless, note that the throughput difference is smaller when comparing the throughput for small and large channel correlations in case of the constant scheduling weights in comparison to the optimal, adaptive choice of the scheduling weights per frame (Figure~\ref{BL_capa_perframeweights2}).

We conclude the discussion regarding the quasi-convexity by summarizing the following guidelines for choosing the scheduling weights:   
Firstly, in general the optimal weights are lower than $0.5$ even for channels with high correlation coefficients.  
Secondly, in comparison to a weak channel correlation, a strong one allows us to set a relatively bigger scheduling weight. 
Thirdly, it is important to have an accurate characterization of the channel correlation, otherwise the FBL throughput can be significantly reduced. 
In particular, the optimal scheduling scheme is more sensitive to an inaccurate knowledge of the channel correlations than the constant weight heuristic.
Furthermore, a low error probability constraint leads to a small feasible set for choosing the scheduling weights. 
Finally, it appears that for constant weight scheduling, the choice of the weights is less sensitive to wrong choices, especially if these choices end up being too large. 
In the case of the optimal scheduling, this only applies to cases where one of the link weights is set optimally.

\subsection{Optimal vs. Constant Scheduling}
In this subsection, we focus on investigating the performance gap between the two schemes presented in Section~\ref{sec:OptimalScheduling} and~\ref{sec:Fixed} for a set of variable parameters with respect to the SNR, reliability constraint and channel correlation coefficient.
To start with, we show the FBL throughput of the two schemes versus the average channel SNR while considering different settings of the reliability threshold $\varepsilon_{\rm th}$.
\begin{figure}[!t]
\centering{\includegraphics[width=100mm, trim=10 10 0 8]{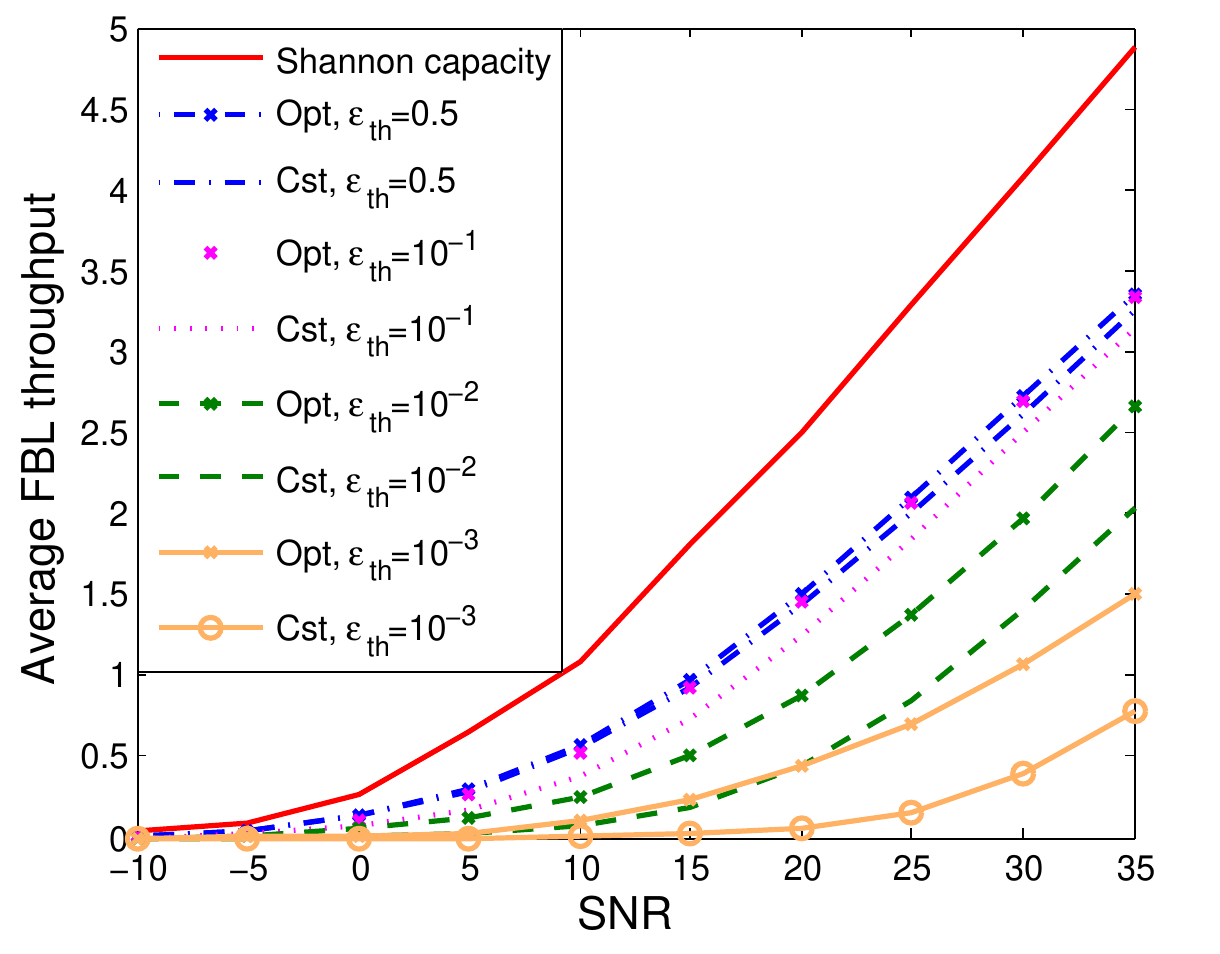}}
\caption{Comparison of the average FBL throughput [bit/ch.use] for the two different schemes (optimal and constant scheduling) versus the average channel SNR for different reliability constraints $\varepsilon_{\rm th}$. For the channel correlations, we considered $\rho^2_1=0.7$ and $\rho^2_2=0.5$.}
\label{BL_capa_threshold}
\end{figure}
The results are shown in Figure~\ref{BL_capa_threshold}, where the average FBL throughputs are based on the optimal/sub-optimal choice of coding rate under either the optimal scheduling or the best constant heuristic.
Firstly, we observe that  the lower the reliability threshold is, the lower the optimal average FBL throughput is. 
Secondly, a higher SNR also increases the gap between the optimal scheduling and the constant heuristic.
Lastly, a lower reliability constraint $\varepsilon_{\rm th} $ leads to a significantly bigger gap between the two schemes. 
For instance, the gap is quite big under the constraint $\varepsilon_{\rm th} =10^{-3}$  while it is small when $\varepsilon_{\rm th} =0.5$. 
This suggests that it only pays off to spend the additional computational complexity for the optimal scheduling scheme in case of a high reliability constraint (i.e. a rather low requirement on the error probability).
In case of a rather low reliability constraint, there is no big difference between the two scheduling schemes.
This is essentially due to the fact that in case of a high reliability constraint the constant heuristic needs to select a proportionally lower value for the scheduling weight to fulfill the reliability constraint even in cases where the instantaneous channel state drops significantly below the outdated CSI.
In case of the optimal scheduling, this can be compensated for by frame-specific scheduling of the weights.

Finally, we show in Figure~\ref{Dynamic_fixed_blt}  the  average FBL throughput of the two different schemes regarding different channel correlation coefficients while the reliability constraint is fixed at $\varepsilon_{\rm th}=10^{-2}$.
\begin{figure}[!t]
\centering{\includegraphics[width=100mm, trim=10 10 0 8]{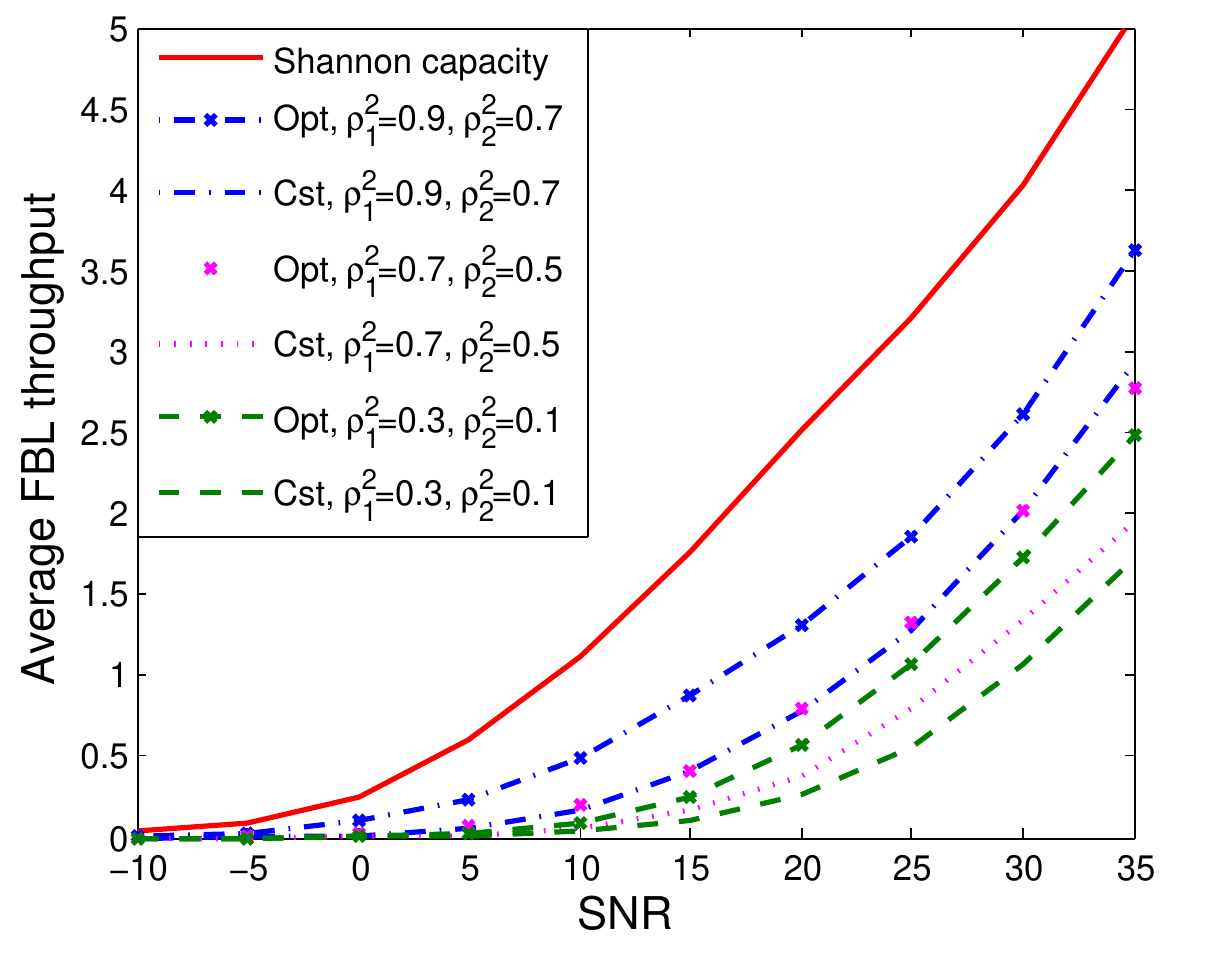}}
\caption{Comparison of the average FBL throughput [bit/ch.use] for the two different schemes (optimal and constant scheduling) versus the average channel SNR for different correlation coefficients of the links. For the reliability constraint we set $\varepsilon_{\rm th}=10^{-2}$.}
\label{Dynamic_fixed_blt}
\end{figure}
We observe that there is a big loss in comparison to the Shannon capacity, even for the FBL throughput with a strong channel correlation.
In addition, a stronger channel correlation introduces a higher FBL throughputs for both the optimal scheduling and the constant heuristic. 
Furthermore, in the high SNR region the performance gap between the two schemes is less influenced by the channel correlation, e.g., at point SNR=25 dB the gap between the two schemes of the case $(\rho^2_1=0.9,\rho^2_2=0.7)$ is quite similar to the gaps under the other two cases.
This is due to the fact that a strong channel correlation makes the outdated CSI more accurate, which reduces the importance of choosing the best scheduling weight. 
As a result, the performance gap between the optimal and the heuristic scheduling schemes is relatively constant.
On the other hand, in the low SNR region the gap between the two schemes is slightly bigger for the case with a strong channel correlation.
 
Combining the insights from Figure~\ref{BL_capa_threshold} and Figure~\ref{Dynamic_fixed_blt}, we can conclude that the performance gap between the proposed optimal and heuristic scheduling schemes mainly depends on the error probability threshold and the average channel SNR,  while it is only marginally influenced by the channel correlation coefficients.
Thus, it is perhaps only worth to spend the computational complexity of the optimal scheme in case of high reliability constraints and a rather high average SNR.

\section{Conclusion}
\label{sec:Conclusion}
In this work, we study the finite blocklength performance of relaying with outdated CSI. 
Both an optimal and an low-complexity sub-optimal scheduling scheme are proposed to maximize the FBL throughput while satisfying a reliability constraint regarding the data transmission. 
We show that in both cases an optimal performance can be obtained by exploiting the quasi-concavity of the FBL throughput with respect to scheduling weights.
By numerical analysis, we conclude a set of guidelines for the design of efficient relaying systems in the FBL regime.
Firstly, 
it is important to have accurate channel correlation information, otherwise the inaccurate channel correlation coefficients can reduce the throughput. 
In particular, the optimal scheme is more sensitive regarding the accuracy of the channel correlation.
Secondly, the optimal scheme is more sensitive to erroneous selection of the scheduling weights in comparison to the constant scheme. 
Thus, in practice a precise computation of the scheduling weights in case of the optimal scheme needs to be performed which nevertheless only pays off in certain scenarios.
For the constant scheme, a less accurate computation of the optimal weights leads already to a satisfactory performance in particular if the scheduling weights are chosen rather too large versus too small. 
Thirdly, the performance gap between the proposed two schemes depends mainly on the reliability constraint regarding the data transmissions. 
The stricter this constraint is, the more does the optimal scheme outperform the constant scheme. 
Finally, the performance gap between the two schemes is less influenced by the channel correlation coefficients.

\appendices
\section{Proof of   Proposition  1}
Based on~\eqref{eq:Overall_error_pro_overchannelfading} and~\eqref{eq:BL_capacity_of_a_frame}, we immediately have ${\mu}_{{\rm{FBL}},i}  ={\mathcal{C}_{{\rm{FBL}}}}({r_i}) =  {{(1 - {{\bar \varepsilon }_{1,i}}({r_i}))(1 - {{\bar \varepsilon }_{2,i}}({r_i})){r_i}}}/{2}$. 
Therefore, the first and second derivatives of ${\mathcal C_{{\rm{FBL}}}}$ with respect to $r_i$ are given by:
\begin{equation}
\begin{split}
\frac{{\partial {\mathcal{C}_{{\rm{FBL}}}}}}{{\partial {r_i}}} = & \frac{{(1 - {{\bar \varepsilon }_{1,i}}({r_i}))(1 - {{\bar \varepsilon }_{2,i}}({r_i}))}}{2} \\
&- \frac{{\partial {{\bar \varepsilon }_{1,i}}}}{{\partial {r_i}}}\frac{{(1 - {{\bar \varepsilon }_{2,i}}({r_i})){r_i}}}{2} - \frac{{\partial {{\bar \varepsilon }_{2,i}}}}{{\partial {r_i}}}\frac{{(1 - {{\bar \varepsilon }_{1,i}}({r_i})){r_i}}}{2},
\end{split}
\end{equation}
  
\begin{equation}
\label{eq: FBL throughput_second_div}
\begin{split}
\frac{{{\partial ^2}{\mathcal{C}_{{\rm{FBL}}}}}}{{{\partial ^2}{r_i}}}  =   
&- \frac{{\partial {{\bar \varepsilon }_{1,i}}}}{{\partial {r_i}}}(1 - {{\bar \varepsilon }_{2,i}}({r_i})) - \frac{{\partial {{\bar \varepsilon }_{2,i}}}}{{\partial {r_i}}}(1 - {{\bar \varepsilon }_{1,i}}({r_i})) \\
& - \frac{{{\partial ^2}{{\bar \varepsilon }_{1,i}}}}{{{\partial ^2}{r_i}}}\frac{{(1 - {{\bar \varepsilon }_{2,i}}({r_i})){r_i}}}{2} - \frac{{{\partial ^2}{{\bar \varepsilon }_{2,i}}}}{{{\partial ^2}{r_i}}}\frac{{(1 - {{\bar \varepsilon }_{1,i}}({r_i})){r_i}}}{2}\\
& + \frac{{\partial {{\bar \varepsilon }_{1,i}}}}{{\partial {r_i}}}\frac{{\partial {{\bar \varepsilon }_{2,i}}}}{{\partial {r_i}}}{r_i}.
\end{split}
\end{equation}

In the following, we prove  {Proposition~1} by showing $\frac{{{\partial ^2}{\mathcal C_{{\rm{FBL}}}}}}{{{\partial ^2}{r_i}}}<0$.
 
Recall that $ \varepsilon _{k,i}= {\mathop{\mathcal{P}}} (\gamma_{k,i},r_i,m),k=1 ,2$. According to~\eqref{eq:single_link_data_rate}, we have:
\begin{equation}
\label{eq: overallerror_first_div}
\frac{{\partial {\varepsilon _{k,i}}}}{{\partial {r_i}}} = \frac{{{m^{\frac{1}{2}}}{\exp{( - \frac{{m{{\left( {{\mathcal{C}}({\gamma _{k,i}}) - r_i} \right)}^2}}}{{2( {1 - {2^{ - 2{\mathcal{C}}({\gamma _{k,i}})}}} ){{\left( {{{\log }_2}e} \right)}^2}}})}}}}{{\sqrt {2\pi } {{\left( {1 - {2^{ - 2{\mathcal{C}}({\gamma _{k,i}})}}} \right)}^{\frac{1}{2}}}{{\log }_2}e}} > 0,
\end{equation}
\begin{equation}
\label{eq: overallerror_second_div}
\frac{{{\partial ^2}{\varepsilon _{k,i}}}}{{{\partial ^2}{r_i}}} = \frac{{{m^{\frac{3}{2}}}\left( {{\mathcal{C}}({\gamma _{k,i}}) - {r_i}} \right){\exp{( - \frac{{m{{\left( {{\mathcal{C}}({\gamma _{k,i}}) - {r_i}} \right)}^2}}}{{2( {1 - {2^{ - 2{\mathcal{C}}({\gamma _{k,i}})}}} ){{\left( {{{\log }_2}e} \right)}^2}}})}}}}{{\sqrt {2\pi } {{( {1 - {2^{ - 2{\mathcal{C}}({\gamma _{k,i}})}}} )}^{\frac{3}{2}}}{{\left( {{{\log }_2}e} \right)}^3}}}.
\end{equation}

Based on~\eqref{eq:single_expected_error_prob}, we have:
\begin{equation}
\frac{{\partial {{\bar \varepsilon }_{k,i}}}}{{\partial {r_i}}} = \int\nolimits_0^{ + \infty } {\mathop
\mathbb{P}\left[ \gamma _{k,i} | \hat{\gamma} _{k,i}\right] 
 \frac{{\partial {\varepsilon _{k,i}}}}{{\partial {r_i}}}d} {\gamma _{k,i}} > 0,
\end{equation}
\begin{equation}
\label{eq:2nd_div_r_i}
\frac{{{\partial ^2}{{\bar \varepsilon }_{k,i}}}}{{{\partial ^2}{r_i}}} = \int\nolimits_0^{ + \infty } {\mathop\mathbb{P}\left[ \gamma _{k,i} | \hat{\gamma} _{k,i}\right]\frac{{{\partial ^2}{\varepsilon _{k,i}}}}{{{\partial ^2}{r_i}}}d} {\gamma _{k,i}}.
\end{equation} 
 
%
As ${{\eta _k}} \le {{\rho^2 _k}}, k =1,2$,  the following inequality holds: ${r_i} < {\mathop\mathcal{C}\nolimits} \left( {\mathop {\min }\nolimits_{k = 1,2} \{ {\eta _k}{{\hat \gamma }_{k,i}}\} } \right) \le {\mathop\mathcal{C}\nolimits} \left( {\mathop {\min }\nolimits_{k = 1,2} \{ {\rho^2 _k}{{\hat \gamma }_{k,i}}\} } \right)$. Hence, ${{\mathcal{C}}({\gamma _{k,i}}) - {r_i}}>0 $ and therefore $\frac{{{\partial ^2}{\varepsilon _{k,i}}}}{{{\partial ^2}{r_i}}} >0$   during the intervals ${\gamma _{k,i}} \in [{{\eta _k}{{\hat \gamma }_{k,i}}}, {\rho^2 _k}{{\hat \gamma }_{k,i}} )$ and ${\gamma _{k,i}} \in [{\rho^2 _k}{{\hat \gamma }_{k,i}} + \infty )$.
As $\mathop\mathbb{P}\left[ \gamma _{k,i} | \hat{\gamma} _{k,i}\right]>0$, hence we have: $\int\nolimits_{{\eta _k}{{\hat \gamma }_{k,i}}}^{{\rho^2 _k}{{\hat \gamma }_{k,i}}} {\mathop\mathbb{P}\left[ \gamma _{k,i} | \hat{\gamma} _{k,i}\right]\frac{{{\partial ^2}{\varepsilon _{k,i}}}}{{{\partial ^2}{r_i}}}d} {\gamma _{k,i}} > 0$ and $\int\nolimits_{{\rho^2 _k}{{\hat \gamma }_{k,i}}}^{ + \infty } {\mathop\mathbb{P}\left[ \gamma _{k,i} | \hat{\gamma} _{k,i}\right]\frac{{{\partial ^2}{\varepsilon _{k,i}}}}{{{\partial ^2}{r_i}}}d} {\gamma _{k,i}} > 0$.

Considering equation~\eqref{eq:Relationship_half_probability}, we have:
 \begin{equation}
\label{eq:20}
\begin{split}
\!\!&\int\nolimits_{{\rho^2 _k}{{\hat \gamma }_{k,i}}}^{ + \infty } \!\!\!\!\!{\mathop\mathbb{P}\left[ \gamma _{k,i} | \hat{\gamma} _{k,i}\right]\frac{{{\partial ^2}{\varepsilon _{k,i}}}}{{{\partial ^2}{r_i}}}d} {\gamma _{k,i}} \\
 = &\int\nolimits_{{\rho^2 _k}{{\hat \gamma }_{k,i}}}^{ + \infty } \! \!\!\!\!\!{\mathop\mathbb{P}\left[ \gamma _{k,i} | \hat{\gamma} _{k,i}\right]\frac{{{m^{\frac{3}{2}}}\left( {{\mathcal{C}}({\gamma _{k,i}}) \!- \!{r_i}} \right){e^{ - \frac{{m{{( {{\mathcal{C}}({\gamma _{k,i}}) \!-\! {r_i}} )}^2}}}{{2( {1 - {2^{ - 2{\mathcal{C}}({\gamma _{k,i}})}}} ){{\left( {{{\log }_2}e} \right)}^2}}}}}}}{{\sqrt {2\pi } {{( {1 \!-\! {2^{ - 2{\mathcal{C}}({\gamma _{k,i}})}}} )}^{\frac{3}{2}}}{{\left( {{{\log }_2}e} \right)}^3}}}d} {\gamma _{k,i}}\\
 >& \int\nolimits_0^{{\rho^2 _k}{{\hat \gamma }_{k,i}}}  \!\!\!\!\!\!\!{\mathop\mathbb{P}\left[ \gamma _{k,i} | \hat{\gamma} _{k,i}\right]\frac{{{m^{\frac{3}{2}}}\left| {{\mathcal{C}}({\gamma _{k,i}}) \!-\! {r_i}} \right|{e^{ - \frac{{m{{( {{\mathcal{C}}({\gamma _{k,i}}) \!- \!{r_i}} )}^2}}}{{2( {1 - {2^{ - 2{\mathcal{C}}({\gamma _{k,i}})}}} ){{\left( {{{\log }_2}e} \right)}^2}}}}}}}{{\sqrt {2\pi } {{( {1 \!- \!{2^{ - 2{\mathcal{C}}({\gamma _{k,i}})}}} )}^{\frac{3}{2}}}{{\left( {{{\log }_2}e} \right)}^3}}}d} {\gamma _{k,i}}\\
 >& \int\nolimits_0^{{\eta _k}{{\hat \gamma }_{k,i}}} \!\!\!\!\!\!\!\! {\mathop\mathbb{P}\left[ \gamma _{k,i} | \hat{\gamma} _{k,i}\right]\frac{{{m^{\frac{3}{2}}}\left| {{\mathcal{C}}({\gamma _{k,i}}) \!- \!{r_i}} \right|{e^{ - \frac{{m{{( {{\mathcal{C}}({\gamma _{k,i}}) \!- \!{r_i}} )}^2}}}{{2( {1 \!- \!{2^{ - 2{\mathcal{C}}({\gamma _{k,i}})}}} ){{\left( {{{\log }_2}e} \right)}^2}}}}}}}{{\sqrt {2\pi } {{( {1 - {2^{ - 2{\mathcal{C}}({\gamma _{k,i}})}}} )}^{\frac{3}{2}}}{{\left( {{{\log }_2}e} \right)}^3}}}d} {\gamma _{k,i}}\\
 =& \left| {\int\nolimits_0^{{\eta _k}{{\hat \gamma }_{k,i}}} \!\!\!\!\!\! {\mathop\mathbb{P}\left[ \gamma _{k,i} | \hat{\gamma} _{k,i}\right]\frac{{{\partial ^2}{\varepsilon _{k,i}}}}{{{\partial ^2}{r_i}}}d} {\gamma _{k,i}}} \right|.
\end{split}
\end{equation}

So far, it has been shown that 
 \begin{equation}
\label{eq:21}
\begin{split}
\frac{{{\partial ^2}{{\bar \varepsilon }_{k,i}}}}{{{\partial ^2}{r_i}}}   =& \int\nolimits_0^{ + \infty }      {\mathop\mathbb{P}\left[ \gamma _{k,i} | \hat{\gamma} _{k,i}\right]\frac{{{\partial ^2}{\varepsilon _{k,i}}}}{{{\partial ^2}{r_i}}}d} {\gamma _{k,i}} \\ 
\ge &   \int\nolimits_{{\eta _k}{{\hat \gamma }_{k,i}}}^{{\rho^2 _k}{{\hat \gamma }_{k,i}}}    {\mathop\mathbb{P}\left[ \gamma _{k,i} | \hat{\gamma} _{k,i}\right]\frac{{{\partial ^2}{\varepsilon _{k,i}}}}{{{\partial ^2}{r_i}}}d} {\gamma _{k,i}} \\
&+ \int\nolimits_{{\rho^2 _k}{{\hat \gamma }_{k,i}}}^{ + \infty }   {\mathop\mathbb{P}\left[ \gamma _{k,i} | \hat{\gamma} _{k,i}\right]\frac{{{\partial ^2}{\varepsilon _{k,i}}}}{{{\partial ^2}{r_i}}}d} {\gamma _{k,i}}\\
 &- \left| {\int\nolimits_0^{{\eta _k}{{\hat \gamma }_{k,i}}}      {\mathop\mathbb{P}\left[ \gamma _{k,i} | \hat{\gamma} _{k,i}\right]\frac{{{\partial ^2}{\varepsilon _{k,i}}}}{{{\partial ^2}{r_i}}}d} {\gamma _{k,i}}} \right|\\
 >& 0  \ .
\end{split}
\end{equation}

According to~\eqref{eq:single_link_error_pro}, the error probability of a link is higher than 0.5 only if the coding rate is higher than the Shannon capacity.
Recall that  the coding rate chosen by the source satisfies ${r_i} < {\mathop{\mathcal{C}}\nolimits} \left( {\mathop {\min }\limits_{k = 1,2} \{ {\eta _k}{{\hat \gamma }_{k,i}}\} } \right) \le {\mathop{\mathcal{C}}\nolimits} \left( {\mathop {\min }\limits_{k = 1,2} \{ {\rho^2 _k}{{\hat \gamma }_{k,i}}\} } \right)$.  This makes the expected error probability of each single link (during frame $i$) be lower than 0.5, i.e., ${{{\bar \varepsilon }_{k,i}}}<0.5, k=1,2$. 
Based on~\eqref{eq: FBL throughput_second_div}, we have:
\begin{equation}
\begin{split}
\frac{{{\partial ^2}{\mathcal C_{{\rm{FBL}}}}}}{{{\partial ^2}{r_i}}} 
&<  - \frac{{{\partial ^2}{{\bar \varepsilon }_{1,i}}}}{{{\partial ^2}{r_i}}}\frac{{{r_i}}}{4} - \frac{{{\partial ^2}{{\bar \varepsilon }_{2,i}}}}{{{\partial ^2}{r_i}}}\frac{{{r_i}}}{4} + \frac{{\partial {{\bar \varepsilon }_{1,i}}}}{{\partial {r_i}}}\frac{{\partial {{\bar \varepsilon }_{2,i}}}}{{\partial {r_i}}}{r_i}\\
&< \left( {2\frac{{\partial {{\bar \varepsilon }_{1,i}}}}{{\partial {r_i}}}\frac{{\partial {{\bar \varepsilon }_{2,i}}}}{{\partial {r_i}}} - \frac{{{\partial ^2}{{\bar \varepsilon }_{1,i}}}}{{{\partial ^2}{r_i}}} - \frac{{{\partial ^2}{{\bar \varepsilon }_{2,i}}}}{{{\partial ^2}{r_i}}}} \right)\frac{{{r_i}}}{2}\\
&\le \left( {2\frac{{\partial {{\bar \varepsilon }_{1,i}}}}{{\partial {r_i}}}\frac{{\partial {{\bar \varepsilon }_{2,i}}}}{{\partial {r_i}}} - \sqrt {\frac{{{\partial ^2}{{\bar \varepsilon }_{1,i}}}}{{{\partial ^2}{r_i}}}\frac{{{\partial ^2}{{\bar \varepsilon }_{2,i}}}}{{{\partial ^2}{r_i}}}} } \right)\frac{{{r_i}}}{2} \ .
\end{split}
\end{equation}
Hence, $\frac{{{\partial ^2}{\mathcal C_{{\rm{FBL}}}}}}{{{\partial ^2}{r_i}}}<0$ if $\frac{{{\partial ^2}{{\bar \varepsilon }_{k,i}}}}{{{\partial ^2}{r_i}}} - 4{\left( {\frac{{\partial {{\bar \varepsilon }_{k,i}}}}{{\partial {r_i}}}} \right)^2}>0$.

According to the Cauchy–Schwarz inequality, we have ${\left( {\frac{{\partial {{\bar \varepsilon }_{k,i}}}}{{\partial {r_i}}}} \right)^2} = \{ \int\nolimits_0^{ + \infty } {\mathop
\mathbb{P}\left[ \gamma _{k,i} | \hat{\gamma} _{k,i}\right]  \frac{{\partial {\varepsilon _{k,i}}}}{{\partial {r_i}}}d} {\gamma _{k,i}} \}^2 
\le    \!\!\!\int\nolimits_0^{ + \infty }   \!{\mathop \mathbb{P}\left[ \gamma _{k,i}   \!| \hat{\gamma} _{k,i}\right]^2   \!(\!\frac{{\partial {\varepsilon _{k,i}}}}{{\partial {r_i}}}})^2   \!d{\gamma _{k,i}}
  \!<  \!  \int\nolimits_0^{ + \infty }   \!{\mathop \mathbb{P}\left[ \gamma _{k,i}   \!| \hat{\gamma} _{k,i}\right]   \!(\!\frac{{\partial {\varepsilon _{k,i}}}}{{\partial {r_i}}}})^2   \!d{\gamma _{k,i}}.$

 Hence, we have:

$\frac{{{\partial ^2}{{\bar \varepsilon }_{k,i}}}}{{{\partial ^2}{r_i}}} - 4 \left( {\frac{{\partial {{\bar \varepsilon }_{k,i}}}}{{\partial {r_i}}}} \right)^2 \\
> \int\nolimits_0^{ + \infty } {\mathop\mathbb{P}\left[ \gamma _{k,i} | \hat{\gamma} _{k,i}\right][ \frac{{{\partial ^2}{\varepsilon _{k,i}}}}{{{\partial ^2}{r_i}}}-4 (\frac{\partial \varepsilon _{k,i}}{\partial {r_i}})^2 ]} d{\gamma _{k,i}}\\
= \int\nolimits_0^{ + \infty }\frac {m \cdot e^A} {\sqrt {2\pi }  ({\log }_2 e)^2  \left( {1 - {2^{ - 2{\mathcal{C}}({\gamma _{k,i}})}}} \right) } \cdot {\mathop\mathbb{P}\left[ \gamma _{k,i} | \hat{\gamma} _{k,i}\right] B} d{\gamma _{k,i}}
$, 

where 
$A= {( - \frac{{m{{\left( {{\mathcal{C}}({\gamma _{k,i}}) - r_i} \right)}^2}}}{{2( {1 - {2^{ - 2{\mathcal{C}}({\gamma _{k,i}})}}} ){{\left( {{{\log }_2}e} \right)}^2}}})  }$ and $B=   \frac { m^{\frac{1}{2}} \left( {{\mathcal{C}}({\gamma _{k,i}}) - r_i} \right) } { {{\left( {1 - {2^{ - 2{\mathcal{C}}({\gamma _{k,i}})}}} \right)}^{\frac{1}{2}}}{{\log }_2}e}
- \frac {4 e^A}{\sqrt {2\pi }}    <   \frac { m^{\frac{1}{2}} \left( {{\mathcal{C}}({\gamma _{k,i}}) - r_i} \right) -2{\text {ln}} 2 } { {{\left( {1 - {2^{ - 2{\mathcal{C}}({\gamma _{k,i}})}}} \right)}^{\frac{1}{2}}}{{\log }_2}e} $. There exists a positive constant $t$,   which makes $B\le\frac { t \cdot m^{\frac{1}{2}} \left( {{\mathcal{C}}({\gamma _{k,i}}) - r_i} \right)  } { {{\left( {1 - {2^{ - 2{\mathcal{C}}({\gamma _{k,i}})}}} \right)}^{\frac{1}{2}}}{{\log }_2}e}$.   
Same to the discussion in~\eqref{eq:20} and~\eqref{eq:21}, it holds that: 

\noindent$ \int\nolimits_0^{ + \infty }\frac {m \cdot e^A \cdot \mathop\mathbb{P}\left[ \gamma _{k,i} | \hat{\gamma} _{k,i}\right]} {\sqrt {2\pi }  ({\log }_2 e)^2  \left( {1 - {2^{ - 2{\mathcal{C}}({\gamma _{k,i}})}}} \right) }  {  \frac { t \cdot m^{\frac{1}{2}} \left( {{\mathcal{C}}({\gamma _{k,i}}) - r_i} \right)  } { {{\left( {1 - {2^{ - 2{\mathcal{C}}({\gamma _{k,i}})}}} \right)}^{\frac{1}{2}}}{{\log }_2}e}} d{\gamma _{k,i}}>0$. 
Hence,  $\frac{{{\partial ^2}{\mathcal C_{{\rm{FBL}}}}}}{{{\partial ^2}{r_i}}} <0$. As $ {\mu}_{{\rm{FBL}},i}  = {\mathcal C_{{\rm{FBL}}}}(r_i)$,  ${\mu}_{{\rm{FBL}},i} $ is concave in $r_i$.

\section{Proof of   Corollary 1}
The coding rate of frame $i$ is scheduled based on the outdated CSI of the bottleneck link, i.e., $\min\{\eta_{1,i}\hat\gamma_{1,i},\eta_{2,i}\hat\gamma_{2,i}\}$. According to~\eqref{eq:single_link_data_rate}, we know that $r_i$ is strictly increasing in  $\min\{\eta_{1,i}\hat\gamma_{1,i},\eta_{2,i}\hat\gamma_{2,i}\}$. 
Hence, when the source schedules the coding rate $r_i$, high values of $\eta_{1,i}$ and $\eta_{2,i}$ lead to a big  $r_i$. In other words, $r_i$ is monotonically increasing in $\eta_{k,i}, k =1,2$.
\\
$\Rightarrow$ $\forall$  $x_k <y_k$,  $x_k, y_k \in (0,\rho_k]$ and $  \lambda_k   \in [0,1]$, we have ${\left. r_i \right|_{\eta_{k,i} = x_k}} < {\left. r_i \right|_{\eta_{k,i} = \lambda_k x_k + \left( {1 - \lambda_k } \right)y_k}} < {\left. r_i \right|_{\eta_{k,i} = y_k}}$, where $ k = 1,2$.
\\
 Based on Proposition 1, ${\mathcal{C}_{{\rm{FBL}}}}$ is concave in $r_i$,
\\
$\Rightarrow$ $\min \left\{ {{\mathcal{C}_{{\text{BL}},i}}\left( {{{\left. r_i \right|}_{\eta_{k,i} = x_k}}} \right),{\mathcal{C}_{{\text{BL}},i}}\left( {{{\left. r_i \right|}_{\eta_{k,i} = y_k}}} \right)} \right\} \leqslant {\mathcal{C}_{{\text{BL}},i}}\left( {{{\left. r_i \right|}_{\eta = \lambda_k x_k + \left( {1 - \lambda } \right)y}}} \right)$.
\\
$\Rightarrow$  ${\mathcal{C}_{{\text{BL}},i}}$ is quasi-concave in $\eta_{k,i}$, where $k=1,2$ and $0 < \eta_{k,i}  \le \rho^2_k$.

\section{Proof of   Proposition 2}
According to the proof of Corollary~1,  $r_i$, $i=1,2,...,+\infty$ is monotonically increasing in $\eta_{k}, k =1,2$.

$\Rightarrow$ $\forall$  ${x_{k,i}} <{y_{k,i}}$,  $i=1,2,...,+\infty$, $x_{k,i}, y_{k,i} \in (0,\rho_k]$ and $  \lambda_{k,i}   \in [0,1]$, we have ${\left. r_i \right|_{\eta_{k} = x_{k,i}}} < {\left. r_i \right|_{\eta_{k} = \lambda_{k,i} x_{k,i} + \left( {1 - \lambda_{k,i} } \right)y_{k,i}}} < {\left. r_i \right|_{\eta_{k} = y_{k,i}}}$, where $ k = 1,2$.

As shown in Proposition~1, ${\mathcal{C}_{{\rm{FBL}}}}$ is concave in $r_i$. 
Hence, ${\mathcal{C}_{{\rm{FBL}}}} = \sum\limits_i {{\mathcal{C}_{{\rm{FBL}}}}} $ is concave in ${\bf{r}}=(r_1,r_2,...,r_i,...)$.

$\Rightarrow$ $\min \left\{ {\sum\limits_i {{\mathcal{C}_{{\rm{FBL}}}}} \left( {{{\left. {{r_i}} \right|}_{{\eta _k} = {x_{k,i}}}}} \right),\sum\limits_i {{\mathcal{C}_{{\rm{FBL}}}}} \left( {{{\left. {{r_i}} \right|}_{{\eta _k} = {y_{k,i}}}}} \right)} \right\} \le \sum\limits_i {{\mathcal{C}_{{\rm{FBL}}}}} \left( {{{\left. {{r_i}} \right|}_{{\eta _k} = \lambda_{k,i} {x_{k,i}} + \left( {1 - {\lambda _{k,i}}} \right){y_{k,i}}}}} \right)$.
\\
$\Rightarrow$  ${\mathcal{C}_{{\text{BL}}}}$ is quasi-concave in $\eta_{k}$, where $0 < \eta_{k}  \le \rho^2_k$ and $k=1,2$.

\bibliographystyle{IEEEtran}

\begin{thebibliography}{10}
\providecommand{\url}[1]{#1}
\csname url@samestyle\endcsname
\providecommand{\newblock}{\relax}
\providecommand{\bibinfo}[2]{#2}
\providecommand{\BIBentrySTDinterwordspacing}{\spaceskip=0pt\relax}
\providecommand{\BIBentryALTinterwordstretchfactor}{4}
\providecommand{\BIBentryALTinterwordspacing}{\spaceskip=\fontdimen2\font plus
\BIBentryALTinterwordstretchfactor\fontdimen3\font minus
  \fontdimen4\font\relax}
\providecommand{\BIBforeignlanguage}[2]{{%
\expandafter\ifx\csname l@#1\endcsname\relax
\typeout{** WARNING: IEEEtran.bst: No hyphenation pattern has been}%
\typeout{** loaded for the language `#1'. Using the pattern for}%
\typeout{** the default language instead.}%
\else
\language=\csname l@#1\endcsname
\fi
#2}}
\providecommand{\BIBdecl}{\relax}
\BIBdecl


\bibitem{Laneman_2004}
J.~Laneman, D.~Tse, and G.~W. Wornell, ``Cooperative diversity in wireless
  networks: Efficient protocols and outage behavior,'' \emph{IEEE Trans. Inf. Theory}, vol.~50, no.~12, pp. 3062--3080, Dec.  2004.

\bibitem{Karmakar_2011}
S.~Karmakar and M.~Varanasi, ``The diversity-multiplexing tradeoff of the
  dynamic decode-and-forward protocol on a MIMO half-duplex relay channel,''
\emph{IEEE Trans. Inf. Theory}, vol.~57, no.~10, pp.
  6569--6590, Oct. 2011.

\bibitem{Li_2014}
J.~Li, B.~Makki and T.~Svensson, "Performance analysis and cooperation mode switch in HARQ-based relaying."  in \emph{IEEE Globecom Workshops}. Austin, TX,  Dec. 2014.

\bibitem{Wendong_2011}
 F. Parzysz, M. Vu and F. Gagnon, ``Impact of propagation environment on energy-efficient relay placement: Model and performance analysis,'' \emph{IEEE   Trans. Wireless Commun.}, vol.~13, no.~4, pp. 2214--2228,
  Feb. 2011.


\bibitem{Yulin_2011}
Y.~Hu and L.~Qiu, ``A novel multiple relay selection strategy for LTE-advanced
  relay systems,'' in \emph{Proc.  IEEE Vehicular Technology Conference (Spring)},
  Budapest, Hungary, May 2011.



\bibitem{Hu_2015_effective_capacity}
Y.~Hu,  J.~Gross and A.~Schmeink,  ``QoS-Constrained energy efficiency of cooperative ARQ in multiple DF relay systems,'' \emph{IEEE Trans. Veh.Technol.}, vol.~65, no.~2, pp.848--859, Feb. 2016



\bibitem{Bhatnagar_2013}
M.~Bhatnagar, ``On the capacity of decode-and-forward relaying over Rician
  fading channels,'' \emph{IEEE Commun. Lett.}, vol.~17, no.~6, pp.
  1100--1103, Jun. 2013.

\bibitem{Verdu_2010}
Y.~Polyanskiy, H.~Poor, and S.~Verdu, ``Channel coding rate in the finite
  blocklength regime,'' \emph{IEEE Trans. Inf. Theory},
  vol.~56, no.~5, pp. 2307--2359,  May 2010.

\bibitem{Polyanskiy_2011}
------, ``Dispersion of the Gilbert-Elliott channel,''  \emph{IEEE Trans. Inf. Theory}, vol.~57, no.~4, pp. 1829--1848, Apr. 2011.


\bibitem{Yang_2014}
W.~Yang, G.~Durisi, T.~Koch and Y.~Polyanskiy ``Quasi-static multiple-antenna fading
  channels at finite blocklength,'' \emph{IEEE Trans. Inf. Theory}, vol.~60,
  no.~7, Jul. 2014.

\bibitem{Gursoy_2013}
M. C. Gursoy, ``Throughput analysis of
  buffer-constrained wireless systems in the finite blocklength regime,''
\emph{ EURASIP J. Wireless Commun.}, vol. 2013:290, Dec. 2013.

\bibitem{Gursoy_2013_2}
G. Ozcan and M. C. Gursoy, “Throughput of cognitive radio systems with finite blocklength
codes,” \emph{IEEE J. Sel. Areas Commun.}, vol. 31, no. 11, pp. 2541-2554, Nov. 2013.



\bibitem{Peng_2011}
P.~Wu and N.~Jindal, ``Coding versus ARQ in fading channels: How reliable
  should the phy be?'' \emph{IEEE Trans. Commun.}, vol.~59,
  no.~12, pp. 3363--3374, Dec. 2011.

\bibitem{Makki_2014}
B.~Makki, T.~Svensson, and M.~Zorzi, ``Finite block-length analysis of the
  incremental redundancy HARQ,''  \emph{IEEE Wireless Commn. Lett.},
  vol.~3, no.~5, pp. 529--532, Oct. 2014.



%
\bibitem{Xu_2016}
S. Xu, T. H. Chang, S. C. Lin, C. Shen and G. Zhu, "Energy-efficient packet scheduling with finite blocklength codes: Convexity analysis and efficient algorithms,"  \emph{IEEE Trans. Wireless. Commn.}, vol.15, no.8, pp.5527-5540, Aug. 2016.
\bibitem{Hu_2015}
 Y.~Hu, J.~Gross and A.~Schmeink, ``On the capacity of relaying with finite
  blocklength,'' \emph{IEEE Trans. Veh. Technol.},  vol.~65, no.~3, pp. 1790--1794, Mar. 2016.

\bibitem{Hu_letter_2015}
-----, ``On the performance advantage of relaying under
  the finite blocklength regime,'' \emph{IEEE Commn. Lett.}, vol.~19, no.~5,
  pp. 779 - 782, May 2015.


\bibitem{Hu_TWC_2015}
 Y.~Hu,  A.~Schmeink and J.~Gross ``Blocklength-limited performance of relaying under quasi-static
  Rayleigh channels,'' \emph{IEEE Trans. Wireless. Commn.}, vol.~15, no.~7, pp. 4548 - 4558, Jul. 2016.


\bibitem{Li_2016_ISIT}
Y. Li, M. C. Gursoy and S. Velipasalar, "Throughput of two-hop wireless channels with queueing constraints and finite blocklength codes,"  \emph{IEEE  ISIT}, Barcelona, Spain, Jul.  2016.


\bibitem{Hyadi_TVT_2015}
A. Hyadi, M. Benjillali and M. S. Alouini, "Outage performance of decode-and-forward in two-way relaying with outdated CSI,"\emph{IEEE Trans. Veh. Technol.}, vol. 64, no. 12, pp. 5940-5947, Dec. 2015.

\bibitem{Jiang_TWC_2016}
W. Jiang, T. Kaiser and A. J. H. Vinck, "A robust opportunistic relaying strategy for co-operative wireless communications," \emph{IEEE Trans. Wireless. Commn.}, vol. 15, no. 4, pp. 2642-2655, Apr. 2016.

\bibitem{Michalopoulos_2016}
D. S. Michalopoulos, J. Ng and R. Schober, "Optimal relay selection for outdated CSI," in IEEE Communications Letters, vol. 17, no. 3, pp. 503-506, Mar. 2013.

\bibitem{Fei_IET_2016}
L. Fei, J. Zhang, Q. Gao and X. H. Peng, "Outage-optimal relay strategy under outdated channel state information in decode-and-forward cooperative communication systems," \emph{IET Commn.}, vol. 9, no. 4, pp. 441-450, Mar.   2015.

\bibitem{Vicario_2006}
J. Vicario and C. Antón-Haro, “Analytical assessment of multi-user
vs. spatial diversity trade-offs with delayed channel state information,"
 \emph{IEEE Commun. Lett.}, Aug. 2006.

\bibitem{Mallik_2003}
R.~Mallik, ``On multivariate Rayleigh and exponential distributions,''
 \emph{IEEE Trans. Inf. Theory}, vol.~49, no.~6, pp.
  1499--1515, Jun. 2003.

\bibitem{Bessel_2005}
G. Arfken, H. Weber,"Mathematical Methods for Physicists,"  \emph{Academic Press}, 6th edition, Jul. 2005.

\bibitem{Molisch_2011}
A.~F. Molisch, ``Wireless Communications,'' \emph{IEEE Press - Wiley}, 2011.

\end{thebibliography}

\end{document}